%% file: main__1_.tex
\documentclass[10pt]{article}
\usepackage[utf8]{inputenc}
\input{macro-gen-prev.tex}
\input{macro-this-prev.tex}
\usepackage{fullpage}
\title{\textbf{$\Oh(\log ^{*}n)$-approximation for $k$-Center in Sublinear MPC}
}
\title{{\bf On Parallel $k$-Center Clustering}}
{
\author{Sam Coy\footnote{University of Warwick, UK} \and  Artur Czumaj\footnotemark[1] \and Gopinath Mishra\footnotemark[1]}
\author{\textbf{Sam Coy}\thanks{Research supported in part by the Centre for Discrete Mathematics and its Applications (DIMAP),
by an EPSRC studentship, and 
by the Simons Foundation Award No 663281 granted to the Institute of Mathematics of the Polish Academy of Sciences for the years 2021-2023.
}
\hspace{4mm} 
\textbf{Artur Czumaj}\thanks{Research supported in part by the Centre for Discrete Mathematics and its Applications (DIMAP), by EPSRC award EP/V01305X/1, by a Weizmann-UK Making Connections Grant, by an IBM Award, and by the Simons Foundation Award No. 663281 granted to the Institute of Mathematics of the Polish Academy of Sciences for the years 2021--2023.}
\hspace{4mm} 
\textbf{Gopinath Mishra}\thanks{Research supported in part by the Centre for Discrete Mathematics and its Applications (DIMAP), by EPSRC award EP/V01305X/1, and by the Simons Foundation Award No 663281 granted to the Institute of Mathematics of the Polish Academy of Sciences for the years 2021-2023.}
\\[0.05in]
Department of Computer Science \\
Centre for Discrete Mathematics and its Applications 
\\
University of Warwick
}

\author{\textbf{Sam Coy}\thanks{E-mail: S.Coy@warwick.ac.uk. Department of Computer Science, University of Warwick, UK.
Research supported in part by the Centre for Discrete Mathematics and its Applications (DIMAP),
by an EPSRC studentship, and 
by the Simons Foundation Award No 663281 granted to the Institute of Mathematics of the Polish Academy of Sciences for the years 2021-2023.}
\hspace{4mm}
\textbf{Artur Czumaj}\thanks{E-mail: A.Czumaj@warwick.ac.uk. Department of Computer Science and Centre for Discrete Mathematics and its Applications, University of Warwick, UK. Research supported in part by the Centre for Discrete Mathematics and its Applications (DIMAP), by EPSRC award EP/V01305X/1, by a Weizmann-UK Making Connections Grant, by an IBM Award, and by the Simons Foundation Award No. 663281 granted to the Institute of Mathematics of the Polish Academy of Sciences for the years 2021--2023.}
\hspace{4mm}
\textbf{Gopinath Mishra}\thanks{E-mail: Gopinath@nus.edu.sg. School of Computer Science, National University of Singapore. Research supported in part by the Centre for Discrete Mathematics and its Applications (DIMAP), by EPSRC award EP/V01305X/1, and by the Simons Foundation Award No 663281 granted to the Institute of Mathematics of the Polish Academy of Sciences for the years 2021-2023.}\\
University of Warwick
}

\author{\textbf{Sam Coy}\thanks{E-mail: S.Coy@warwick.ac.uk. Department of Computer Science, University of Warwick, UK. Research supported in part by the Centre for Discrete Mathematics and its Applications (DIMAP), by an EPSRC studentship, and by the Simons Foundation Award No. 663281 granted to the Institute of Mathematics of the Polish Academy of Sciences for the years 2021--2023.}\\
University of Warwick
\and
\textbf{Artur Czumaj}\thanks{E-mail: A.Czumaj@warwick.ac.uk. Department of Computer Science and Centre for Discrete Mathematics and its Applications, University of Warwick, UK. Research supported in part by the Centre for Discrete Mathematics and its Applications (DIMAP), by EPSRC award EP/V01305X/1, by a Weizmann-UK Making Connections Grant, by an IBM Award, and by the Simons Foundation Award No. 663281 granted to the Institute of Mathematics of the Polish Academy of Sciences for the years 2021--2023.}\\
University of Warwick
\and
\textbf{Gopinath Mishra}\thanks{E-mail: Gopinath@imsc.res.in. Theoretical Computer Science Group, The Institute of Mathematical Sciences, HBNI, Chennai, India.  Research was supported in part  by EPSRC award EP/V01305X/1, and by the Simons Foundation Award No 663281 granted to the Institute of Mathematics of the Polish Academy of Sciences for the years 2021-2023.} \\
The Institute of Mathematical Sciences
}
}
\date{}

\newcount\shortyear\newcount\shorthour\newcount\shortminute
\shorthour=\time\divide\shorthour by 60\shortyear=\shorthour
\multiply\shortyear by 60\shortminute=\time\advance\shortminute by
-\shortyear \shortyear=\year\advance\shortyear by -1900

\def\zeit{\number\shorthour:\ifnum\shortminute<10 0\number\shortminute
\else\number\shortminute\fi}

\makeatletter

\begin{document}

\maketitle
\input{abstract.tex}

\section{Introduction}
\label{sec:intro}

Clustering large data is a fundamental primitive extensively studied because of its numerous applications in a variety of areas of computer science and data science. It is a central type of problem in modern data analysis, including the fields of data mining ~\cite{han2001data}, pattern recognition \cite{jain1999data}, machine learning \cite{xu2005survey}, networking and social networks \cite{girvan2002community}, and bioinformatics \cite{eisen1998cluster}. 
In a typical clustering problem, the goal is to partition the input data into subsets (called clusters) such that the points assigned to the same cluster are ``similar'' to one another, and data points assigned to different clusters are ``dissimilar''. The most extensively studied clustering problems are $k$-means, $k$-median, $k$-center, various notions of hierarchical clustering, and also variants of these problems with some additional constraints (e.g., fairness or balance).~\cite{lloyd1982least, arthur2007kmeans++, gonzalez1985kcenter, aggarwal2013data, chierichetti2017fair, bachem2018scalable}.


While originally the clustering problems have been studied in the context of classical sequential computation,
most recently a large amount of research has been devoted to the non-sequential computational settings such as distributed and parallel computing, mainly because these are the only settings capable of performing computations in a reasonable time on large inputs, and because data is frequently collected on different sites and clustering needs to be performed in a distributed manner with low communication.

In this paper, we consider one of the fundamental  
clustering problems, the \emph{$k$-center} problem {in the Euclidean space}, on the \emph{Massively Parallel Computation (MPC)} model. \mpc is a modern theoretical model of parallel computation, inspired by frameworks such as MapReduce \cite{mapreduce}, Hadoop \cite{hadoop}, Dryad \cite{dryad}, and Spark \cite{spark}. Introduced just over a decade ago by Karloff \etal \cite{mpc_introduced} (and later refined, e.g., in \cite{ANOY14,BKS17,CLMMOS18,GSZ11}), the model has been the subject of an increasing quantity of fundamental research in recent years, becoming nowadays the standard theoretical parallel model of algorithmic study.

\mpc is a parallel system with \machines \emph{machines}, each with \lspace words of \emph{local memory}. (We also consider the \emph{global space} \gspace, which is the total space used across all machines, $\gspace = \lspace \cdot \machines$.) Computation takes place in synchronous rounds: in each round, each machine may perform arbitrary computation on its local memory and then exchange messages with other machines. {While the local computation at each machine is not explicitly accounted for, it is typically required to be polynomially bounded.} Each message is sent to a single machine specified by the machine sending the message. Machines must send and receive at most \lspace words each round. The messages are processed by recipients in the next round. At the end of the computation, machines collectively output the solution. 
The goal is to design an \mpc algorithm that solves a given task in as few rounds as possible.

For an input of size $n$, the local memory $s$ available to each machine  should be sublinear in $n$, since if $s \ge n$, a single machine could solve the problem without any communication. Moreover, the total memory across all machines must be at least $n$ to store the input, and ideally should not exceed this by much. In this paper, we focus on the \emph{low-local-space} \mpc setting, where the local space of each machine is strongly sublinear in the input size, i.e., $\lspace = \Oh(n^{\delta})$ for some arbitrarily constant $\delta \in (0,1)$. {The low-local-space regime is especially attractive due to its scalability in large-scale distributed systems, where individual machines typically have limited memory compared to the overall dataset size.} At the same time, this setting is particularly challenging in that it requires extensive inter-machine communication to solve clustering problems for the input data scattered over many machines.



In recent years we have seen a number of very efficient, often constant-time, parallel clustering algorithms that have been relying on a combination of a \emph{core-set} and a ``reduce-and-merge'' approach \cite{BadoiuHI02,HarPeledMazumdar04,Bahmani12,Indyk14}, where one gradually filters the data set by typically reducing its size on every machine to $\tOh(k)$ until all the data can be stored on a single machine, at which point the problem is solved locally. Observe that this approach has an inherent bottleneck that requires that any machine must be able to store $\Omega(k)$ data points. Intuitively, this follows from the fact that if a machine sees $k$ data points that are all very far away from each other, it needs to keep track of all $k$ of them, for otherwise it might miss all the information about one of the clusters, which in turn could lead to a large miscalculation of the objective value. Similar arguments could be also used to argue that each machine needs to communicate $\Omega(k)$ points to the others. Similar arguments suggest that each machine must communicate $\Omega(k)$ points (see \cite{CSWZ16} for a formalization for $k$-center, $k$-median, and $k$-means under a worst-case partition in the coordinator model. However, this lower bound does not directly extend to \mpc, where inputs are typically randomly partitioned; moreover, with $\Oh(1)$ rounds, an adversarial partition can be transformed into a random one in \mpc).
Because of that, most of the earlier clustering \mpc algorithms, especially those working in a constant number of rounds (see, e.g., \cite{coreset_kcenter1,coreset_kcenter2}), require $\Omega(k)$ or even $\Omega(k) \cdot n^{\Omega(1)}$ local space. Therefore in the setting considered in this paper, of \mpc with local space per machine of $\lspace = \mathcal{O}(n^{\delta})$, the approach described above cannot be applied when the number of clusters is large, when $k = \omega(\lspace)$. This naturally leads to the main challenge in the design of clustering algorithms for \mpc with low-local-space: \emph{how to efficiently partition the data into $k$ good quality clusters on an \mpc with local space $\lspace \ll k$}. For example, when $k \le \lspace$, a single machine can store $\Omega(k)$ candidate centers or summaries, enabling standard coreset-based or sampling approaches to work. In contrast, when $k \gg \lspace$, no machine can even store one representative per cluster, making such approaches infeasible and necessitating fundamentally different techniques.



\medskip

In this paper, we focus on the \emph{$k$-center clustering problem} in the Euclidean space, a standard, widely studied, and widely used formulation of metric clustering. The problem is, given a set of $n$ input points, to find a subset of size $k$ of these points (called \emph{centers}) such that that maximum distance of a point to its nearest center is minimized. Specifically, in this work, we focus on the case where $k \gg \lspace$ and hence, when $k$ is quite large relative to $s$: one can think of the $k$-center clustering problem as ``compressing'' the input set of $n$ points into $k$ points.
Very recently, this problem has been addressed in the Euclidean space $\R^d$ for constant $d$ by Bateni \etal~\cite{bateni-kcenter}, who showed  one can design an $\Oh(\log\log n)$-round \mpc algorithm, with local space $\lspace=\Oh(n^{\delta})$ and global space  $\gspace=\tOh(n^{1+\delta})$,\footnote{$\tOh(f)$ hides a polynomial factor in $\log f$.} that returns an $\Oh(\log\log\log n)$-approximate solution with $k+o(k)$ centers. Our main result is an improved bound in the \mpc model:


\begin{theo}[\textbf{Main result}]
\label{theo:inf}\label{thm:main-informal}
There exists an $\Oh(\log\log n)$ rounds  \mpc algorithm that computes an $\Oh(\log^* n)$ approximate solution to the $k$-centers problem using $k(1+o(1))$ centers. Moreover, the \mpc algorithm has local space $\lspace = \Oh(n^\delta)$ and global space $\gspace = \tOh(n^{1+\rho})$ for any constant $\rho \in (0,\delta]$. The $n$ input points are in $\R^d$ for some constant $d$. Our algorithm succeeds with high probability. Here the constant inside $\Oh(\log^* n)$ depends on $\rho$ and that in $\Oh(\log \log n)$ depends on $\delta$.

\begin{rem}
   Our results, together with those of \cite{bateni-kcenter}, primarily focus on designing fully scalable \mpc algorithms for \kcen in constant-dimensional Euclidean space, particularly when $k > \lspace$. To simplify the presentation, we follow the approach of \cite{bateni-kcenter} and do not explicitly present the dependence of $\rho$, $\delta$, and $d$ in the approximation factor, the number of rounds, or the global space. 
\end{rem}
%
\end{theo}

{The algorithmic framework is based on a repeated application of \emph{locally sensitive sampling} introduced by Bateni \etal~\cite{bateni-kcenter}: sampling a set of ``hub'' points by sampling each of the (remaining) points independently with a suitable probability. Then we assign all other points to a nearby hub using \emph{locally sensitive hashing}, and then adding new hubs to such that the distance of each point to the nearest hub is bounded by a parameter.
We improve the approximation factor of $\Oh(\log \log \log n)$ due to Bateni \etal~\cite{bateni-kcenter} by a careful examination of the progress of clusters in some fixed optimal clustering over the course of the algorithm.
 We achieve the improvement in the approximation factor with a deeper understanding of random process of sampling and assigning, with an optimal clustering, that is, how the set of points in an cluster of an optimal clustering behaves when we sample, select hubs, and assign points to hubs. Ideally, one wants to have exactly one point from each cluster. We show that, with suitable probability, the size of the clusters decreases at a desirable rate over iterations. However, clusters may not satisfy the  size constraint with high probability, and a key challenge in our analysis is to carefully bound the sizes of those clusters that do not meet this requirement.} Additionally, we provide a more flexible guarantee on the global space, providing an accuracy parameter $\rho \in (0,\delta]$, which can be set to reduce global space used at the expense of a larger approximation ratio (or vice versa).
 


A preliminary version of this work appeared in the Proceedings of SPAA 2023 \cite{CoyCM23}. The current article includes the complete proofs of several claims that were stated without proof in the conference version due to space constraints. Beyond filling in these details, the presentation has been revised for clarity and completeness, making this version self-contained.
\subsection{Related work}
\label{subsec:related-work}

There has been a large amount of work on various variants of the clustering problems (see, e.g., \cite{clustering_survey} for a survey of research until 2005), including some extensive study of the $k$-center clustering problem. The $k$-center problem is well known to be NP-hard and simple algorithms are known to achieve a 2-approximation \cite{DF85,gonzalez1985kcenter,HS85}; this approximation ratio is tight unless $\mathrm{P} = \mathrm{NP}$ \cite{HN79}.

The study of clustering in the context of parallel computing is extremely well-motivated: as the size of typical data sets continues to increase, it becomes infeasible to store input data on a single machine, let alone iterate over it many times (as greedy sequential algorithms require (see, e.g., \cite{gonzalez1985kcenter})). It comes therefore as no surprise that there has been a considerable amount of work on $k$-center clustering algorithms in \mpc. In particular, several constant-round, constant-approximation algorithms in the \mpc setting were given recently for general metric $k$-center clustering, see, e.g., \cite{coreset_kcenter1,coreset_kcenter2,coreset_kcenter3}. Much of this work used \emph{coresets} or similar techniques as a means of approximating the structure of the underlying data, naturally implying a requirement that the local space satisfies $\lspace = \Omega(k)$ per machine and global space is $\gspace = \Omega(nk)$ or $\Omega(n^{\epsilon} k^2)$ \cite{coreset_kcenter1,coreset_kcenter2,coreset_kcenter3}. Specifically, Ene \etal \cite{coreset_kcenter1} gave a $\Oh(1)$ round $10$-approximation \mpc algorithm that uses local space $\lspace = \Oh(k^2 n^{\Theta(1)})$, Malkomes \etal \cite{coreset_kcenter2} obtained a $2$-round $4$-approxi\-ma\-tion \mpc algorithm with local space $\lspace = \Omega(\sqrt{nk})$, and Ceccarello \etal \cite{coreset_kcenter3} obtained a $2$-round $(2+\eps)$-approximation \mpc algorithm that uses local space $\lspace = \Oh (\sqrt{nk})$ for the problem in metric spaces with constant doubling dimension. Very recently, Haqi and Zadeh~\cite{HaqiZ23} gave a $\Oh(1)$-round $(2+\eps)$-approximation \mpc algorithm that uses local space $\lspace = \Oh (\sqrt{nk})$ for the problem in constant dimensional Euclidean space. As mentioned earlier, these algorithms are not scalable if $k$ is large relative to $n$ (for example, when $k = n^{1/3}$), making the case of large $k$ particularly challenging. Furthermore, as argued by Bateni \etal \cite{bateni-kcenter}, the case of large $k$ appears naturally in some applications of $k$ clustering, including label propagation used in semi-supervised learning, or same-meaning query clustering for online advertisement or document search \cite{WLFXY09}. Unfortunately, we do not know of any $O(1)$-round, $O(1)$-approximation \mpc algorithm that would use local space $\lspace = o(k)$.

In order to address the case of large $k$, Bateni \etal \cite{bateni-kcenter} considered a relaxed version of $k$-center clustering for low dimensional Euclidean spaces with constant dimension. The goal of that work was to design a scalable \mpc algorithm for the $k$-center clustering problem with a \emph{sublogarithmic number of rounds} of computation, \emph{sublinear space per machine}, and small global space. Bateni \etal \cite{bateni-kcenter} showed that in $\Oh(\log\log n)$ rounds on an \mpc with $\lspace = \Oh(n^\delta)$, one can compute an $\Oh(\log\log\log n)$-approximate solution to constant-dimension Euclidean $k$-center with $k(1 + o(1))$ centers. Their algorithm uses $\tOh(n^{1 + \delta} \cdot \log \Delta)$ global space. Bateni \etal \cite{bateni-kcenter} complemented their analysis by some empirical study to demonstrate that the designed algorithm performs well in practice.

Finally, in the related PRAM model of parallel computation Blelloch and Tangwonsan gave a $2$-approximation algorithm for \kcen \cite{pram-k-center}. However, their algorithm requires $\Omega(n^2)$ processors and it is therefore difficult to translate the approach to our setting.


\subsection{Technical contributions}
\label{sec:diff-intro}

Our main result in \Cref{thm:main-informal} is an extension of the approach developed in Bateni \etal \cite{bateni-kcenter} that significantly improves the quality of the approximation guarantee. To present these two results in the right context, we will briefly describe the main differences between these two works at a high level.

The approach of Bateni \etal \cite{bateni-kcenter} starts with the entire point set $P$ as a set of potential centers (solution), and refines it to
$P = P_0 \supseteq \dots \supseteq  P_{\tau}$, until $\size{P_\tau} = k+o(k)$.
Note that $P_i$ is the set of leftover points  after round $i$, that is, set of potential centers after round $i$. The final set $P_\tau$ is reported as the output. It is not difficult to see that if we consider an optimal clustering $\cC^*$ for $P$, that is, the clustering induced by an optimal solution to the $k$-center problem on $P$ where each point is assigned to its closest center with ties broken arbitrarily, then for every cluster $C \in \cC^*$, the number of potential centers (or leftover points) in $C$ decreases over the rounds. In particular, for every round $i$, we have $|P_{i+1} \cap C| \le |P_i \cap C|$.
Let us define a cluster $C \in \cC^*$ to be \emph{irreducible from round $i$}, if $i$ is the minimum index such that $\size{C \cap P_i} \leq 1$.
Two central properties of the cluster refinement due to Bateni \etal \cite{bateni-kcenter} are that after $\Oh(\log\log n)$ rounds the number of leftover points in each $C \in \cC^*$ reduces to $\tOh(\log n)$, and that after that, the total number of the leftover points in the
reducible
clusters in $\cC^*$ reduces after each round by a constant factor, implying that another $\Oh(\log\log n)$ rounds suffice to ensure that the desired number of centers.  Since the total number of clusters in $\cC^*$ is at most $k$, the number of centers in the irreducible clusters is at most $k$. Moreover, the number of leftover points in the reducible clusters is at most $k \cdot \tOh(\log n)$ after the first $O(\log \log n)$ rounds. In the next $O(\log \log n)$ rounds, this contributes at most $k$ centers from the irreducible clusters and an additional $o(k)$ centers from the reducible clusters.
Hence setting the number of rounds as $\tau = \Oh(\log\log n)$ is sufficient. This is then complemented by the analysis of the quality of the refinements which guarantees that each new refinement adds an additive term of $\Oh(\opt)$ to the cost of the solution such that each point in $P_i$ has a corresponding point at a distance at most $\Oh(\opt)$ in $P_{i+1}$, giving in total a double logarithmic approximation ratio. They also provided a sketch of the analysis to obtain an approximation ratio of $\Oh(\log\log\log n)$, which we believe can be proved.

In our paper we substantially improve the approximation factor to $\Oh(\log^* n)$ by extending the framework in the following sense. We show that, after $\Oh(\log\log n)$ rounds, the size of each cluster in $\cC^*$ reduces to $\tOh(\log n)$ such that the refinement in each round adds an additive error of $\Oh(\frac{\opt}{\log\log n})$ to the cost of the solution. Then, we show that after additional $\Oh(\log\log\log n)$ rounds, the sizes of \emph{almost all} (but not all) clusters in $\cC^*$ reduce to a $\tOh(\log\log n)$ such that the refinement in each round adds an additive error of $\Oh(\frac{\opt}{\log\log\log n})$ to the cost of the solution. Next, we show that after another $\Oh(\log\log\log\log n)$ rounds, the sizes of \emph{almost all} clusters in $\cC^*$ reduce to a $\tOh(\log\log\log n)$ such that the refinement in each round adds $\Oh(\frac{\opt}{\log\log\log\log n})$ to the cost of the solution, and so on. We continue this until the sizes of almost all clusters in $\cC^*$ reduce to $\Oh(\log^*n)$. Observe that the total number of rounds taken so far is bounded by $\Oh(\log\log n)$, and we can argue that the current solution has an approximation ratio of $\Oh(\log^*n)$. An important challenge in analyzing this approach is that \emph{not all clusters satisfy these size guarantees with high probability}. Indeed, we cannot obtain a high probability guarantee by cluster refinement relying on random sampling of the already small clusters; we can ensure only that \emph{most of the clusters are getting small}.

Let $\cC^{**} \subseteq \cC^*$ be the clusters that satisfy the reduction property as discussed above, that is, such that the number of points in each cluster of $\cC^{**}$ is bounded by $\Oh(\log^* n)$ currently. We argue that the total number of points in the reducible clusters in $\cC^{**}$ reduces by a constant factor after each successive round, adding an additive error of $\Oh(\opt)$ each time. This implies that another $\Oh(\log (\log^* n))$ rounds are good enough to ensure that we have the desired number of centers at the end. To bound the total number of centers, we also need to show that the number of centers in clusters in $\cC^{*}\setminus \cC^{**}$ (that is, the set of clusters which fail to adhere to a size guarantee at some point during the algorithm) is bounded. Note that we cannot \emph{track} which clusters succeed or fail (doing so would require us to know an optimal clustering), and so we use $\cC^*$ and $\cC^{**}$ only for the analysis. In summary, the approach sketched above will reduce the number of clusters to $k(1+o(1))$, and will ensure that the total number of rounds spent by our algorithm is $\Oh(\log\log n)$ and the approximation ratio of our solution is $\Oh(\log^*n)$. A more detailed overview is in \Cref{sec:over}.

Our approach, like that of Bateni \etal \cite{bateni-kcenter}, relies heavily on locality-sensitive hashing (LSH). 
When constructing $P_{i+1} \subseteq P_i$, the algorithm samples each point in $P_i$ independently with an appropriate probability. 
Each sampled point is added to a hub. 
Then, each remaining point is assigned to one of the hubs using LSH.
 We provide a detailed implementation of LSH in \mpc  in \Cref{sec:hub}.

\subsection{Notation and preliminaries}

We now introduce the notation used through the paper.

First, we present the setting of the parameters of our \mpc. The \kcen algorithm in this paper works for any local space $\lspace = \Oh(n^\delta)$ for a constant $0 < \delta < 1$: the setting of $\delta$ has only a constant factor impact on the running time. Similarly, the \mpc can have any global space $\gspace = \widetilde{\Oh}(n^{1+\rho})$ for some constant $\rho > 0$: $\rho$ can be made arbitrarily small, and its setting has a constant factor impact on the approximation ratio. We sometimes refer to \mpc with these choices of \lspace and \gspace simply as ``\specifiedmpc'' in the rest of this paper.

Let us recall that certain operations, particularly sorting and prefix sum of $n$ elements, and broadcasting a value of size $<\lspace$, can be computed deterministically in $O(1)$ rounds (see \cite{GSZ11}).

The input to our problem is a set $P$ of $n$ points in $\mathbb{R}^d$, where $d$ is a constant, and an integer parameter $k < n$. We define $d(p, q)$ as the Euclidean distance between points $p$ and $q$ in $\mathbb{R}^d$. We generalize this notation to the distance between a point and a set: $d(p, S) := \min\limits_{q \in S} d(p, q)$ is the minimum distance from $p$ to a point in $S$. We define $\cost(P,S) := \max\limits_{p \in P} d(p,S)$ as the distance of the point in $P$ which is ``furthest away'' from any point in $S$. 

Let $\Delta_{\min}$ and $\Delta_{\max}$ denote the minimum and maximum distance between any two points in the input set $P$, respectively, and define $\Delta := \Delta_{\max} / \Delta_{\min}$. We assume that $\Delta_{\min}$ is known. Without loss of generality, we rescale the input so that the minimum distance between any two points in $P$ is $1$, and thus $\Delta$ equals the maximum distance between any two points in $P$. We also assume that $\Delta$ is known to the algorithm. If $\Delta$ is unknown, we can compute an estimate $\Delta'$ satisfying $\Delta \leq \Delta' \leq 2\Delta$, as follows, which is sufficient for our purposes. First, select an arbitrary point $p \in P$, then find the point $q \in P$ farthest from $p$, and then find the point $r \in P$ farthest from $q$. We output $\Delta' := 2 \cdot d(q, r)$. This value satisfies $\Delta \leq \Delta' \leq 2\Delta$, thus giving a 2-approximation to $\Delta$. This procedure can be implemented in $O(1)$ rounds in the MPC model: broadcasting the chosen point to all machines, local computation of distances by each machine, and global aggregation to identify the farthest point can all be done in constant rounds. Hence, the overall algorithm completes in $O(1)$ rounds.


We denote the set $\{1,\ldots,t\}$ by $[t]$ and $\log^{(i)} n := \underbrace{\log\dots\log}_{i} n$ the iterated logarithm of $n$. By convention $\log^{(0)} n := n$. The notations $\widetilde{\Oh}(f)$ and $\widetilde{\Theta}(f)$ hide polynomial factors in $\log f$.

We now formally define the $k$-center clustering problem.

\begin{defi}[\textbf{Clustering}]
Let $P$ be a set of points in $\R^d$. A \emph{clustering} $\cC$ of $P$ is a partition of $P$ into nonempty clusters $C_1,\ldots,C_t$. The \emph{radius} of cluster $C_i$ is $\min\limits_{x \in C_i} \max\limits_{y \in C_i}d(x,y)$, and the \emph{cost} of the clustering $\cC$ is the maximum of the radii of the clusters $C_1,\ldots,C_t$.
\end{defi}

\begin{defi}[\textbf{$k$-center clustering problem}]
Let $k,n, d \in \N$ with $k \leq n$, and $P$ be a set of  points in $\R^d$. The \kcen problem for $P$ is to find a set  $S^* \subseteq P$ such that $$S^*=\argmin\limits_{S \subseteq P: \size{S}=k} \cost(P,S).$$
Moreover, $\cost(P,S^*)$ is defined as the (optimal) cost of the \kcen problem for $P$.
\end{defi}

Observe that the clustering induced by the optimal solution to $k$-center for $P$ is,
among all clusterings of $P$ with $k$ centers, the one with the minimum cost.

\begin{defi}[\textbf{Optimal clustering with a given cost}]\label{def:Cr} Let $P$ be a set of points in $\R^d$. For $r \in \N$, an \emph{ optimal clustering of $P$ with cost $r$}, denoted by $\cC_r$, is clustering of $P$ with minimum number of clusters among all clustering of $P$ whose cost is at most $r$. We write $\size{\cC_r}$ for the number of clusters in $\cC_r$. 

\end{defi}


\subsection{Our results --- detailed bounds}
\label{sec:results}

We now present in details the main result of this paper:

\begin{theo}[{\bf Main result}]
\label{theo:main-star}
Let $P$ be any set of $n$ points in $\R^d$ for some constant $d$ and let \opt denote the optimal cost of the $k$-center clustering problem for $P$.
There exists an \mpc algorithm that in ${\Oh}(\log\log n)$ rounds determines with high probability a set $T \subseteq P$ of $k+o(k)$ centers, such that $\cost(P,T) = \Oh(\log^*n) \cdot \opt$.
The \mpc uses local space $\lspace = \Oh(n^{\delta})$ and global space $\gspace = \tOh(n^{1+\rho} \cdot \log^2\Delta)$ for any constant $\rho \in (0,\delta]$.
%
%
%

\end{theo}

\Cref{theo:main-star} follows directly from a more general theorem.

\begin{theo}[{\bf Generalization of \Cref{theo:main-star}}]
\label{theo:main}
Let $\alpha$ be an arbitrary integer, $1 \le \alpha \le \log^*n - c_0$ for some suitable constant $c_0$. Let $P$ be any set of $n$ points in $\R^d$ for some constant $d$ and let \opt denote the optimal cost of the $k$-center clustering problem for $P$.
There exists an \mpc algorithm that in ${\Oh}(\log\log n)$ rounds determines with high probability a set $T \subseteq P$ of centers, such that $\cost(P,T) = \Oh((\alpha + \log^{(\alpha + 1)}n)) \cdot \opt$ and $\size{T} \le k \cdot \left(1+\frac{1}{\widetilde{\Theta}(\log^{(\alpha)} n)}\right) + \widetilde{\Theta}((\log ^{(\alpha)} n)^3) + {\widetilde{\Oh}(\log n)}$.
The \mpc uses local space $\lspace = \Oh(n^{\delta})$ and global space $\gspace = \tOh(n^{1+\rho} \cdot \log^2\Delta)$ for any constant $\rho \in (0,\delta]$.
%
\end{theo}

Observe that in \Cref{theo:main}, we have $\size{T} = k + o(k)$ whenever $k = \Omega((\log n)^c)$ for some suitable constant~$c$. For \emph{small} values of $k$, we can instead apply one of the earlier algorithms (e.g., due to Ene \etal~\cite{coreset_kcenter1}), as also noted by Bateni \etal~\cite{bateni-kcenter}.

\begin{proof}[Proof of \Cref{theo:main-star}]
    Let $\alpha_0$ be the solution to the equation $\alpha = \log^{(\alpha + 1)} n$; observe that  $\alpha_0 = \Theta(\log^*n)$. Setting $\alpha = \alpha_0$ in \Cref{theo:main}, we get the desired result.
\end{proof}

\Cref{theo:main} can be seen as a fine-grained version of \Cref{theo:main-star}: as $\alpha$ increases the cost of the solution decreases and number of center increases (with the number of rounds always being $O(\log\log n)$). Therefore \Cref{theo:main} is more amiable in practical scenarios in the following sense: $\alpha$ in \Cref{theo:main} can be set to trade off between the quality of the solution and the number of centers in the solution. We would also like to highlight that the result of Bateni \etal \cite{bateni-kcenter} is a special case of \Cref{theo:main} when $\alpha=1$ and $\alpha=2$ to obtain $\Oh(\log\log n)$ and $\Oh(\log\log\log n)$ approximation, respectively.


\subsection{Organization of the paper}
\label{subsec:organization}

In \Cref{sec:over} we give a proof of our main result predicated on the correctness of our main algorithm, and then give an overview of the subroutines which our main algorithm contains. In \Cref{sec:hub} we explain how LSH (locality-sensitive hashing) on \mpc can be implemented to assign each point $p \in P$ to a hub in $H \subseteq P$ which is within a constant factor of the closest hub to $p$. In Sections \ref{sec:samp}~and~\ref{sec:uni} we prove critical properties of subroutines used in our main algorithm, and then in \Cref{sec:main} we prove the correctness of our main algorithm. Finally, \Cref{sec:conclude} contains some conclusions. 
\section{Technical overview}
\label{sec:over}

Recall that \Cref{theo:main-star} is our main result and \Cref{theo:main} is its parameterized generalization. Our proof of \Cref{theo:main} (and hence of \Cref{theo:main-star}) relies on the following main technical theorem.


\begin{theo}[{\bf Main technical theorem proved in this paper}]
\label{theo:main1}
Let $\alpha$ be an arbitrary integer, $1 \le \alpha \le \log^*n - c_0$ for some suitable constant $c_0$.
Let $r$ be an arbitrary positive real.
Let $P$ be any set of $n$ points in $\R^d$ for some constant $d$ and {let $\cC_r $ be the optimal clustering of $P$ with cost at most $r$}. There exists an \mpc algorithm \malg(Algorithm~\ref{algo:main}) that with probability at least $1-\frac{1}{(\log^{(\alpha - 1)} n)^{\Omega(1)}}$, in ${\Oh}(\log\log n)$ rounds determines a set $T \subseteq P$ of centers, such that $\cost(P,T) = \Oh(r \cdot (\alpha+\log^{(\alpha + 1)}n))$ and
$\size{T} \le \size{\cC_r} \cdot \left(1 + \frac{1}{\widetilde{\Theta}(\log^{(\alpha)} n)} \right) + \widetilde{\Theta}((\log^{(\alpha)}n)^3) + {\widetilde{\Oh}(\log n)}$.
The \mpc uses local space $\lspace = \Oh(n^{\delta})$ and global space $\gspace = \tOh(n^{1+\rho} \cdot \log\Delta)$ for any constant $\rho \in (0,\delta]$.
\end{theo}


Note that the number of centers produced by the algorithm described in \Cref{theo:main1} is approximately the same as the number of centers in the optimal clustering of $P$ with cost at most $r$. This contrasts with the standard clustering setting, where the number of clusters is specified as input, without any information about the cost of the solution. Therefore, if we knew a constant-factor approximation to the optimal cost of the $k$-center problem, then by setting $r$ to this value, in \Cref{theo:main1}, we would obtain the desired solution required in \Cref{theo:main}.
 This naturally suggests to run \malg multiple times in parallel in order to obtain \Cref{theo:main}. Note that the success probability of \Cref{theo:main1} is not high. Hence we first run \malg a suitable number of times in parallel to get an algorithm $\mbox{\malg}'$ whose output and space requirements are same as that of \malg, but the success probability is high. Then we run $\mbox{\malg}'$ for $O(\log\Delta)$ choices of $r$ (starting with $r=\Delta$ and decreasing a constant factor each time) in parallel to get algorithm $\mbox{\malg}''$ (the algorithm of \Cref{theo:main}). Moreover, $\mbox{\malg}''$ reports the output of $\mbox{\malg}'$ for the minimum $r$ for which we get the number of centers equals to $k+o(k)$. In what follows, we prove \Cref{theo:main} formally.

\begin{proof}[Proof of  \Cref{theo:main}]

Let us consider algorithm $\mbox{\malg}'$ that runs $\psi$ instances, $\psi=\Oh(\log (\max \{n, \log \Delta\}))$, of the algorithm \malg in parallel. Let $T(1), \ldots, T(\psi)$ be the outputs of the runs of \malg. $\mbox{\malg}'$  reports $T'=T(i)$ with the minimum cardinality as the output. Therefore by \Cref{theo:main1}, with probability at least $1-\frac{1}{(\max\{n, \log \Delta\})^{\Omega(1)}}$, the following is true for $\mbox{\malg}'$:
\begin{enumerate}
    \item[(i)] the number of rounds spent by $\mbox{\malg}'$ is $\Oh(\log \log n)$ and the global space used by $\mbox{\malg}'$ is $\tOh(n^{1+\rho} \log \Delta)$;
    \item[(ii)] $\cost(P,T')=\Oh(r\cdot (\alpha+\log ^{(\alpha + 1)}n))$; and
    \item[(iii)] $\size{T'}\le \size{\cC_r}(1+\frac{1}{\widetilde{\Theta}(\log ^ {(\alpha)} n )})$ + $\widetilde{\Theta}((\log ^{(\alpha)} n)^3) + {\widetilde{\Oh}(\log n)}$.
\end{enumerate}

Next, consider the following observation about $\mbox{\malg}'$:

\begin{obs}
\label{obs:malg-dash}
Let \opt be the optimal cost to the \kcen problem for $P$. If one runs $\mbox{\malg}'$ with radius parameter $r$ with $r \geq  \opt$, then the number of centers reported by $\mbox{\malg}'$ is at most $k\left(1+\frac{1}{\widetilde{\Theta}(\log^{(\alpha)} n)}\right)+\widetilde{\Theta}((\log ^{(\alpha)} n)^3) + {\widetilde{\Oh}(\log n)}$, with probability at least $1-\frac{1}{(\max\{n,\log \Delta\})^{\Omega(1)}}$.
\end{obs}

\Cref{obs:malg-dash} follows from the bound $\size{\cC_r} \le k$ for $r \geq \opt$.

Now we describe the algorithm algorithm $\mbox{\malg}^{''}$, which is the algorithm corresponding  to \Cref{theo:main}. $\mbox{\malg}^{''}$ runs $\phi=\Oh(\log \Delta)$ instances of $\mbox{\malg}'$, with radius parameters $r(1)=\Delta, r(2)=\frac{\Delta}{2}, r(3)=\frac{\Delta}{4}, \ldots, r({\phi})=\Oh(1)$, in parallel. Let $T'(1), \ldots, T'(\phi)$ be the corresponding outputs of the runs of $\mbox{\malg}'$.  $\mbox{\malg}^{''}$ reports $T''=T'(i)$ as the output such that $\size{T'(i)}\leq k(1+\frac{1}{\widetilde{\Theta}(\log ^ {(\alpha)} n)})+ \widetilde{\Theta}((\log ^ {(\alpha)} n)^3) + {\widetilde{\Oh}(\log n)}$ and $r(i)$ is the minimum. So, the round complexity and space complexity of \Cref{theo:main} follow from the round and space complexity guarantee of $\mbox{\malg}'$, respectively. From the guarantee of algorithm $\mbox{\malg}'$ about the set of centers returned by it, we have $\cost(P,T'')=r(i)\cdot (\alpha+{(\log ^{(\alpha + 1)} n)})$, and $\size{T''}\leq k(1+\frac{1}{\widetilde{\Theta}(\log^{(\alpha)} n)})$ + $\widetilde{\Theta}((\log ^{(\alpha)} n)^3) + {\widetilde{\Oh}(\log n)}$, with probability at least $1-\frac{1}{(\max\{ n, \log \Delta\})^{\Omega(1)}}$. From \Cref{obs:malg-dash}, $r(i) \leq 2 \cdot \opt$ with probability at least $1-\frac{1}{(\max\{n, \log \Delta\})^{\Omega(1)}}$. This yields the proof of the guarantee on the approximation factor and the number of centers of \Cref{theo:main}.
    
\end{proof}


\subsection{Overview of the proof of \Cref{theo:main1}}
\label{subsec:proof-theo:main1}

The idea to prove \Cref{theo:main1} is based on the framework which we call \emph{locally sensitive sampling}.
We generate a set $H \subseteq P$ of points (called \emph{hubs}) by sampling each point in $P$ independently with a \emph{suitable} probability, and assign all other points to one of the hubs based on its \emph{locality}. Let $B_h$ be the \emph{bag of the hub $h$}---the set of points associated to a hub $h \in H$. We run a variation of a well known greedy algorithm \cite{gonzalez1985kcenter} (for \kcen in the sequential setting) for each bag in parallel to find a set of intermediate centers $C_h$ for hub $h$ such that $\cost(B_h,C_h) = \Oh(r)$. We again repeat the procedure by setting $\bigcup_{h \in H} C_h$ as the point set. We continue this process a particular number of times with a particular choice of probability and radius parameters, and report the centers, at that point of time, as the final solution.

This framework was recently used by Bateni \etal \cite{bateni-kcenter} to give an $\Oh(\log\log n)$-round \mpc algorithm with local space $\lspace = \Oh(n^\delta)$ and global space $\gspace = \tOh(n^{1+\delta})$, which computes an $\Oh (\log\log\log n)$-approximate solution to \kcen with $k(1 + o(1))$ centers, with high probability. We extend their framework and generalize the analysis to give an $\Oh(\log^*n)$ approximate solution as stated in \Cref{theo:main-star}. Note that \Cref{theo:main1} takes care of \Cref{theo:main-star} via \Cref{theo:main}.

The algorithm corresponding to \Cref{theo:main1} is \malg (Algorithm~\ref{algo:main} in \Cref{sec:main}). Before describing \malg, we describe and contextualize the three subroutines which it uses (\nns, \sas and \ukc).
The main algorithm of Bateni \etal \cite{bateni-kcenter} uses subroutine {\sc Sample-And-Solve} and {\sc Uniform}-$k$-{\sc
center}. We use analogous subroutines \sas and \ukc in our algorithm corresponding to {\sc
Sample-And-Solve} and {\sc Uniform}-$k$-{\sc center} in Bateni \etal ~\cite{bateni-kcenter}, respectively, to achieve the desired result. But there are some differences which we will discuss when we describe \sas and \ukc. Due to our implementation of \nns, we are able to give a more flexible bound on global space. We can improve the approximation ratio mainly due to generalizing their {\sc Uniform}-$k$-{\sc Center} to \ukc in our case and using sophisticated analysis in our main algorithm that calls \ukc.

Let us discuss first at a high level what these subroutines achieve in the context of the framework of locally sensitive sampling (discussed at the beginning of this section). Intuitively, the purpose of \sas is to sparsify dense regions of points: it samples nodes with a given probability and iteratively adds centers in order to ensure that the cost of the centers remains low. \ukc repeatedly uses \sas: its main purpose is to guarantee that the number of centers in each cluster of some fixed optimal clustering decreases in a certain way over time.


\subsubsection*{\nns ($Q,H$)}

Takes as input a set $Q$ of at most $n$ points and a set of hubs $H \subseteq Q$. For all points $q \in Q \setminus H$, it finds a point $\mbox{close}(q) \in H$ such that $d(q,\mbox{close}(q)) = \Oh( d(q,H))$, with probability at least $1-\frac{1}{n^{\Omega(1)}}$. \nns{} can be implemented in MPC with local space $\lspace=\Oh(n^\delta)$ and global space $\gspace=\tOh(n^{1+\rho} \cdot \log \Delta)$ in $\Oh(1)$ rounds. \nns{} uses \emph{locally sensitive hashing}~\cite{lsh, Charikar02} and its implementation in \mpc. For details on \nns, see \Cref{sec:hub}.

Bateni \etal~\cite{bateni-kcenter} employed locality-sensitive hashing in their algorithm to essentially achieve the same objective as ours. We complement their work by providing a detailed description of its implementation and usage.

\subsubsection*{\sas($Q,p,r$)}

Takes a set $Q$ of at most $n$ points, a sampling parameter $p$, and a radius parameter $r$. It produces some set of centers $S \subseteq Q$ such that $\cost(Q,S)=\Oh(r)$.\footnote{The constant inside $\Oh(\cdot)$ depends on $\rho.$} Importantly, this can be implemented  in an \mpc  with local space $\lspace=\Oh(n^{\delta})$ and global space $\tOh(n^{1+\rho}\cdot \log \Delta)$ in $\Oh(1)$ rounds (\Cref{lem:samp}) as, aside from using \nns{} to assign points to hubs, the computation is all done locally. \sas first samples each point in $Q$ (independently) with probability $p$: let $H \subseteq Q$ be the set of sampled points called hubs. Then \sas calls \nns{} with input point set $Q$ and hub set $H$. After getting $\close(q)$ for each $q \in Q\setminus H$, \sas collects all points $B_h \subseteq Q$ assigned to a hub $h \in H$ (including hub $h$) and selects a set of centers $C_h$ from $B_h$ greedily using a variation of the sequential algorithm of \cite{gonzalez1985kcenter}, such that $\cost(B_h,C_h) =\Oh(r)$. Finally, the algorithm outputs $S=\bigcup_{h \in H}C_h$. However, there is a difficulty to overcome: note that $\size{B_h}$ may be $\omega(n^\delta)$. So $B_h$ may not fit into the local memory of a machine. We show that this can be handled by distributing the points in $B_h$ into multiple machines, duplicating $h$ across all such machines. See \Cref{sec:samp} for more details about \sas.

Algorithm \sas in our paper serves essentially the same purpose as the corresponding algorithm due to Bateni \etal \cite{bateni-kcenter}. The approximation guarantee and number of rounds performed are the same in both cases. However, the global space used by our algorithm \sas is more flexible in the following sense: reducing the value of $\rho$ decreases the amount of global space used by the algorithm (global space used is $\tOh(n^{1 + \rho})\cdot \log \Delta$) while increasing the approximation ratio.

\subsubsection*{\ukc$(V,r,t)$}

Takes a set $V$ of at most $ n$ points, a radius parameter $r$, and an additional parameter $t \leq n.$ It produces a set $S$ of centers, by calling \sas $\tau=\Theta(\log \log t)$ times. $S_{i-1}$ is the input to the $i$-th call and $S_i$ is the output of the $i$-th call: overall we have $S_0=V$ and $S_\tau=S$ (the output of \ukc). The probability and radius parameters to the calls to \sas are set \emph{suitably}. From the guarantees we have from \sas, we have the following guarantee for \ukc ( (\Cref{lem:ukc})): (i)  $\cost(V,S)=\Oh(r\cdot \tau)=\Oh(r \log \log t)$ , and (ii) it can be implemented in an \mpc with local space $\lspace=\Oh(n^{\delta})$ and global space $\gspace=\tOh(n^{1+\rho}\cdot \log \Delta)$ in $\Oh(\log \log t)$ rounds. \ukc guarantees a reduction in cluster sizes in an optimal clustering in the following sense.
Consider  a fixed clustering $\cC_r^t$ of $V$ that has cost at most $r$. For $C \in \cC_r^t$: if  $\size{ C \cap V}\leq t$, then $\size{C \cap S} =\Oh( \log t \cdot (\log \log t)^2)$, with probability at least $1-\frac{1}{t^{\Omega(1)}}$.  This is formally stated in  \Cref{lem:uniform3}: note that this ceases to be high probability when $t \in o(n)$. This guarantee on the size reduction plays a crucial role when proving the number of centers reported by \malg in \Cref{sec:main}. For more details on \ukc, see \Cref{sec:uni}.

Our \ukc is a full generalization of the analagous {\sc Uniform}-$k$-{\sc center} in Bateni \etal~\cite{bateni-kcenter}. In particular, {\sc Uniform}-$k$-{\sc Center} is a special case of our \ukc when $t=n$. This generalization plays a crucial role in the correctness of \malg when we call \ukc multiple times. {\sc Uniform}-$k$-{\sc Center} is not robust enough to be called from \malg multiple times to achieve the desired result.

\subsubsection*{\malg $(P,r)$} 

Takes a set $P$ of $ n$ points and a radius parameter $r$. The algorithm consists of two phases, where {\bf Phase~1} consists of $\alpha$ subphases and {\bf Phase 2} consists of $\beta={\Theta}(\log ^{(\alpha + 1)} n)$ subphases. In the $j$-th subphase of {\bf Phase 1}, that is, in {\bf Phase 1.j}, \malg calls $\mbox{\ukc}(T_{j-1},r_{j-1},t_{j-1})$, where {$ T_j$ is the output of $ \text{Uniform-Center}(T_{j-1}, r_{j-1}, t_{j-1})$}, $T_0=P$, $r_0=\frac{r}{\log \log n}$, $t_0=n$, $t_{j}=\widetilde{\Theta}(\log ^{(j)} n),$ and $r_j=\frac{r}{\log \log  t_{j}}$. Observe that the guarantees of \ukc ensure the following:
\begin{enumerate}
\item[(i)] {\bf Phase 1} can be implemented in an \mpc with local space $\lspace=\Oh (n^\delta)$ and global space $\gspace=\tOh(n^{1+\rho} \cdot \log\Delta)$ in $\sum\limits_{j=1}^{\alpha} \log\log t_{j-1}=\Oh(\log \log n)$ rounds;
\item[(ii)] $\cost(T_{j-1},T_j)=\Oh(r_{j-1} \log \log t_{j-1})=\Oh(r)$ for each $j \in [\alpha]$. Hence, $\cost(P,T_{\alpha}) = \Oh(r\alpha)$.
\end{enumerate}

Now consider {\bf Phase 2} of \malg.

In the $i$-th subphase of {\bf Phase 2}, that is {\bf Phase 2.i}, \malg calls $\mbox{\sas}(T_{\alpha + i-1},\frac{1}{2},r)$, where {$T_{\alpha+i}$ is the output of $ \text{Sample-And-Solve}(T_{\alpha+i-1}, 1/2, r)$ and $T=T_{\alpha + \beta}$ is the final output of \malg.}
From the guarantee of \sas, we have
\begin{enumerate}
\item[(i)] {\bf Phase 2} can be implemented in an \mpc  with local space $\lspace=\Oh(n^\delta)$ and global space $\gspace=\tOh(n^{1+\delta}\cdot \log \Delta)$ in $\Oh(\beta)=\Oh(\log ^{(\alpha +1)}n)$  rounds;
\item[(ii)] $\cost(T_{\alpha+i-1},T_{\alpha +1})=\Oh(r)$ for each $i \in [\beta]$. Hence,
     \begin{align*}
        \cost(P,T) &=
        \cost(P,T_{\alpha+\beta}) =
        \cost(P,T_{\alpha})+\Oh(\beta r)
            \\&=
        \Oh(r\cdot (\alpha + \log^{(\alpha +1)}n))
        \enspace.
     \end{align*}
\end{enumerate}

Combining the guarantees concerning the round complexity, global space and approximation guarantee of {\bf Phase 1} and {\bf Phase 2}, we get the claimed guarantees on round complexity, global space, and approximation factor in \Cref{theo:main1} (see \Cref{lem:main}).

Now, we discuss how we bound the number of centers that \malg outputs, that is, $\size{T}$. Consider an optimal clustering $\cC_r$ of $P$ with cost at most $r$. A cluster $C \in \cC_r$ is said to be \emph{active} (after {\bf Phase 1}) if $\size{C \cap T_j} \leq t_j$ for each $j$ with $1 \leq j \leq \alpha$. We say $C$ is \emph{inactive}, otherwise. Using the guarantee given by \ukc concerning the reduction in cluster sizes, we can show that the total number of centers in $T_\alpha$, that are in inactive clusters, is
%
%
$\Oh\left(\frac{ \size{\cC_r}}{(\log ^{(\alpha) } n)^{\Omega(1)}}\right) + \widetilde{\Oh}(\log n)$,
with probability at least $1-\sum\limits_{i=1}^{\alpha}\frac{1}{t_{i-1}^{\Omega(1)}}$ (see \Cref{lem:inter1}). Note that $T_\alpha$ denotes the set of intermediate centers we have after {\bf Phase 1}. So, for any cluster $C \in \cC_r$ that is active after {\bf Phase 1}, it satisfies $\size{C \cap T_\alpha}\leq t_\alpha =\widetilde{\Theta}{(\log ^{(\alpha)} n)}$. That is, with probability at least $1-\sum\limits_{i=1}^{\alpha}\frac{1}{t_{i-1}^{\Omega(1)}}$, we have the following:
\begin{align*}
    \size{T_\alpha}&\leq
    \size{\cC_r}\cdot t_{\alpha}+\Oh\left(\frac{ \size{\cC_r}}{(\log ^{(\alpha ) } n)^{\Omega(1)}}\right) + \widetilde{\Oh}(\log n)
    \enspace.
\end{align*}

We define an active cluster $C \in \cC_r$ is $i$-large if $\size{C \cap T_{\alpha+i-1}}\geq 2$. We show that the total number of intermediate centers in any $i$-large clusters reduces by a constant factor in {\bf Phase 2.i}, with probability at least $1-\frac{1}{t_{\alpha-1}^{\Omega(1)}}$. Note that the total number of intermediate centers in all active large clusters, just before {\bf Phase 2}, is at most $\size{\cC_r}\cdot t_\alpha =\size{\cC_r}\cdot \widetilde{\Theta}(\log ^{(\alpha)} n) $, and we are executing  $\beta={\Theta}(\log ^{(\alpha +1)} n)$ many sub-phases in {\bf Phase 2}. We can show that the total number of centers in the active large clusters, after {\bf Phase 2}, is at most
%
$\frac{\size{\cC_r}}{\widetilde{\Theta}(\log ^{(\alpha)} n)}+\widetilde{\Theta}((\log^{(\alpha)} n)^3)$,
with probability at least $1-\frac{1}{t_{\alpha -1}^{\Omega(1)}}$ (\Cref{lem:inter2}). Combined with the fact the number of active small clusters can be at most $\size{\cC_r}$ with the bound on number of inactive clusters in {\bf Phase 2}, we have the desired bound on $\size{T}$. Full details of \malg and its analysis are presented in \Cref{sec:main}.

\section{Nearest Hub Search}
\label{sec:hub}
 
Recall that our \nns algorithm takes a set $Q$ of points and a set $H \subseteq Q$ of hubs. For each $q \in (Q \setminus H)$, we want to find a hub $h \in H$ such that the distance between $q$ and $h$ is only a constant-factor more than the distance between $q$ and the closest hub to $q$ in $H$: {informally, $h$ is ``almost'' the closest hub to $q$ in $H$.}

In this section, we use \emph{locally sensitive hashing} (LSH)~\cite{lsh} to implement algorithm \nns ($Q, H$) on \mpc. Our implementation of locally sensitive hashing is parameterizable: by setting the parameter $\rho$ appropriately, one can reduce the global space while increase the approximation ratio, or vice versa.

First, we begin by recalling the definition of locally sensitive hashing, introduced in \cite{lsh}:

\begin{defi}[{\bf Locally sensitive hashing} \cite{lsh}]
Let $r \in \R^+$, $c>1$ and $p_1,p_2 \in(0,1)$ be such that $p_1 > p_2$. A hash family $\cH=\{h:\R^d \rightarrow U\}$ is said to be a $(r,cr,p_1,p_2)$-LSH family if for all $x,y \in \R^d$ the following hold:
\begin{itemize}
\item If $d(x,y)\leq r$, then $\pr_{h \in \cH}(h(x)=h(y)) \geq p_1$;
\item If $d(x,y)\geq c r$, then $\pr_{h \in \cH}(h(x)=h(y)) \leq p_2$.
\end{itemize}
\end{defi}

We next present a proposition regarding the existence of a particular hash family, which will be useful for describing and analyzing
 \nns ($Q, H$) in Algorithm~\ref{algo:nn}.

\begin{pro}[\cite{lsh, Charikar02}]
\label{pre:exist-hash}
Let $r, n \in \mathbb{N}$ and $\rho \in (0,1)$. There exists an explicit
$(r, c_\rho r, (1/n)^\rho, 1/n)$-LSH family, where $c_\rho \ge 1/\rho$.
A hash function can be sampled uniformly at random from this family and stored
using $\Oh(d)$ words.
\end{pro}
The explicit hash function can be implemented using the \emph{arc cosine locality-sensitive hashing} \cite{Charikar02, rashtchian2019lsh}.


In \nns ($Q, H$),  $Q$ is a set of at most $n$ points and $H \subseteq Q$ is the set of hubs. Our objective is to find a hub for each point which is at most some constant factor further away than the nearest hub, rather than finding the hub which is the closest. We do this by making $\log \Delta$ guesses about the distance to the nearest hub, and for each guess trying to find a hub within that distance.

For our $\log \Delta$ guesses of $r$ (the distance to the closest hub), we independently
select $L = \Theta(n^\rho)$ hash functions from a $(r, c_\rho r, (1/n)^\rho, 1/n)$-LSH family
and use them to hash all points, including the hubs. We then group all points that share the
same hash value.

To identify a hub close to each point, note that if $h$ hubs and $m$ points share a hash value,
performing $h \cdot m$ distance checks can be prohibitive if $h \cdot m > \lspace$. We show that
when multiple hubs are mapped to the same hash value, it is sufficient to retain only a
constant number of hubs while ensuring that each point is assigned a hub within a constant
factor of its closest hub. This follows from the choice of the hash function and the properties of LSH. The complete algorithm, \nns($Q, H$), is described in Algorithm~1,
and its correctness is established in \Cref{theo:lhs}.


\begin{lem}[{\bf Guarantee of \nns}]
\label{theo:lhs}
Let $Q$ be a set of at most $n$ points in $\R^d$, $H \subseteq Q$ denote the set of hubs, and $c_\rho$ be a suitable constant depending on $\rho$.  Consider \nns ($Q, H$) (as described in Algorithm~\ref{algo:nn}). It does not report {\sc Fail} with high probability. Moreover,
\begin{enumerate}
    \item[(i)]  For all $q \in Q\setminus H$, finds a hub $\mbox{close}(q) \in H$ such that $d(q, \close(q)) < 2c_\rho \cdot d( q,H)$;
    \item[(ii)] It takes $\Oh(1)$ \mpc rounds with local space $\Oh(n^\delta)$ and $\lspace = O(n^\delta)$ and global space $\gspace = \Oh(n^{1+\rho} \cdot \log^2 n \cdot \log \Delta )$.
\end{enumerate}

\end{lem}

\remove{\begin{obs}
For all $x,y \in \R^d$ the followings  holds with probability at least $1-\frac{1}{n^{\Omega(1)}}$:
\begin{itemize}
    \item $d(x,y)\leq r$, then there exists an $f \in \cF$ such that $f(x)=f(y)$
    \item If $d(x,y)\geq c_\rho r$, then for all $f \in \cF$ we have $f(x)\neq f(y)$.
\end{itemize}
\end{obs}}

\remove{
\paragraph*{Locally sensitive hashing}
Let $Q$ be a set of points in $\R^d$. Also $c>1$ and $p_1,p_2 \in (0,1)$ be the parameters. There is a data structure of size $\tOh\left(n^{1+1/c}\right)$ such that

Let $Q$ be a set of points and $\Delta$ be the diameter of $Q$. Let $H \subset Q$ be a set of hubs. The objective of each point in $\cQ$ to find the nearest point in $H$.}


\begin{algorithm}[h]
\caption{\nns ($Q, H$)}\label{algo:nn}
\KwIn{A set $Q$ of at most $n$ points  and a set of hubs $H \subseteq Q$.}
\KwOut{Either reports {\sc Fail}, or for each point $p \in Q$ report $\mbox{close}(p) \in H$ such that 
$d(p,\mbox{close}(p)) \leq 2c_\rho \cdot d(p,H)$, 
where $c_\rho$ is a constant depending only on $\rho$.}
\Begin{
\For{($i=1$ to $I=\Theta(\log n)$)}
{
\For{(j=0 to $\log \Delta$)}
{
Set $r=2^j$

Take $L=\Theta(n^{\rho})$ many hash function $f_1,\ldots,f_L$ ({independently and uniformly at random}) from a $(r,c_\rho r, (1/n)^{\rho},1/n)$-LSH family.

\For{($\ell=1$ to $L$)}
{
Determine $f_\ell(q)$ for each $q \in Q$.

Find the distance of each $q \in Q$ with at most a constant (say $10$) number of hubs $h \in H$ such that $f_\ell(q)=f_\ell(h)$. If we get such a $h \in H$ such that $d(q,h)\leq c_\rho \cdot r$, then we set $\mbox{close}_{ij\ell}(q)=h$ and {\sc null}, otherwise.

}
Set $\mbox{close}_{ij}(q) =\mbox{{\sc null}}$ if $\mbox{close}_{ij\ell}(q)=\mbox{{\sc null}}$ for all $\ell \in [L]$. Otherwise, set $\mbox{close}_{ij}(q)=\mbox{close}_{ij \ell}(q)$ for some $\ell \in [L]$.

}

Set $\mbox{close}_{i}(q) =\mbox{{\sc null}}$ if $\mbox{close}_{ij}(q)=\mbox{{\sc null}}$ for all $j \in [\log \Delta]$. Otherwise, $\mbox{close}_{i}(q)=\mbox{close}_{ij^*}(q)$ such that $j^*$ is minimum among all $j$ for which $\mbox{close}_{ij}(q)$ is not {\sc null}.
}
If there exists a $q \in Q$ such that $\mbox{close}_i(q)$ is {\sc null} for all $i \in [I]$, then report {\sc Fail}.

Otherwise, set $\mbox{close}(q)=\mbox{close}_i(q)$
 for some $i \in [I]$ such that $\mbox{close}_i(q) \neq \mbox{{\sc null}}.$
 }
\end{algorithm}



From Algorithm~\ref{algo:nn}, note that we repeat a procedure (lines 3--12 that find an almost closest hub with probability $2/3$) $I=\Theta(\log n)$ times, and report the output we get from any of the instances. By \Cref{lem:inter-lhs} (stated below), each point $q \in Q$ finds $\mbox{close}(q)\in H$ satisfying the required property with high probability. This will immediately imply the correctness of \Cref{theo:lhs}. We then discuss the \mpc implementation of \nns.

Note that $\close_i(q)$ (which is either {\sc null} or a point in $H$ such that $d(q,\close_i(q))=O(d(q,H))$) denotes the output of \nns for point $q \in Q \setminus H$ and the instance $i \in [I]$.

\begin{lem}
\label{lem:inter-lhs}
For a particular $q \in Q \setminus H$ and $i \in [I]$, $\close_i(q)\in H$ is not {\sc null} and $d(q,\close_i(q))\leq 2 c_\rho \cdot d(q,H)$, with probability at least ${4/5}$.
\end{lem}

\begin{proof}
        Consider $j^*$ such that $d(q,H) \leq r=2^{j^*} \leq 2 \cdot  d(q,H)$,
        and $q_h \in H$ be such that $d(q,q_h) \leq  \textcolor{blue}{r}.$ As
        each $f_\ell, \ell \in [L],$ is a function chosen from $(r,c_\rho r,
        (1/n)^\rho,1/n)$-LSH family, $\pr(f_\ell(q)=f_\ell(q_h)) \geq
        \frac{1}{n^{\rho}}$. As $L=\Theta(n^\rho)$, there
        exists an $\ell^* \in [L]$ such that
        $f_{\ell^*}(q)=f_{\ell^*}(q_h)$ with probability at least
        $9/10$. But our algorithm may not find this particular
        $q_h$ while considering the hubs $h \in H$ such that
        $f_{\ell^*}(q)=f_{\ell^*}(q_h)=f_{\ell^*}(h)$ (See line 8 of
        \nns). Again, as $f_{\ell^*}$ is chosen
        from $(r,c_\rho r, (1/n)^\rho,1/n)$-LSH family, the expected number
        of hubs $h \in H$, with $d(q,h)> c_\rho r$ but
        $f_{\ell^*}(q)=f_{\ell^*}(h)$, is at most $1$. By Markov's
        Inequality, the probability that the number of such hubs is
        more than $10$ is at most $1/10$. So, with probability at
        least $2/3$, \nns, sets
        $\close_{ij^*\ell^*}=h$ for some $h \in H$ such that
        $d(q,h)\leq c_\rho r$, that is, $d(q,h)\leq 2c_\rho \cdot d(q,H)$. Now
        considering the way we set $\close_{ij}(q)$ from
        $\close_{ij\ell}(q)$'s ($\ell \in [L]$) in line 10 and
        $\close_{i}(q)$ from $\close_{ij}(q)$'s $(0 \leq j \leq \log
        \Delta)$ in line 12, we have that \nns
        sets $\close_i(q)\in H$ such that $d(q,\close_i(q))\leq 2 c_\rho \cdot d(q,H)$ with probability at least $4/5$.
\end{proof}
We next analyze how \textsc{Nearest-Hub-Search} sets the value of $\close(q)$ in lines~14--15 from the values $\close_i(q)$. By \Cref{lem:inter-lhs}, with probability at least $1 - \frac{1}{n^{\Omega(1)}}$, the algorithm returns a non-null $\close(q)$ such that
\[
d\bigl(q, \close(q)\bigr) \le 2c_\rho \cdot d(q, H),
\]
since $I = \Theta(\log n)$.

Applying the union bound over all points in $Q \setminus H$, we conclude that \Cref{theo:lhs} follows from \Cref{lem:inter-lhs}, except for the details of the MPC implementation, which we provide in the next section.



\subsection*{\mpc implementation of \nns }

Recall that, we assume $0 < \rho  \leq \delta$. First, notice that, if we can implement lines 4-10 of \nns in \mpc with local space $\lspace = \Oh(n^\delta)$ and global space $\gspace = \Oh( n^{1+\rho} \cdot \log n)$,  then we can run these lines in parallel for each possible value of $i$ and $j$ (adding a factor of $\Oh(\log n \cdot \log \Delta)$ to the global space). Then the results can be aggregated in $\Oh(1)$ rounds using sorting and prefix sum \cite{GSZ11}.

It suffices then to show that lines 4--10 of \nns can be implemented in the desired rounds and space. {Due to \Cref{pre:exist-hash} and using the fact that $0 < \rho \leq \delta$}, the $\Theta(n^\rho)$ number of hash functions in line 5 can be generated locally by some ``leader'' machine and broadcast to the other machines in $O(1)$ rounds. We again perform lines 6--9 in parallel, giving each {$f_\ell~ (\ell \in [L])$ its own set of machines to use.}

We next consider the implementation of lines 7--8 given a specific $\ell \in [L]$. Machines can compute $f_\ell$ locally and without communication. Each point $q\in Q$ is now represented by a tuple $(f_\ell(q), \text{hub}(q), q)$, where $\text{hub}(q) = 1$ if $q \in H$ and $0$, otherwise. Machines then sort these tuples lexicographically and remove (using prefix sum) all but $10$ hubs for each value in $\range(f_\ell)$.\footnote{{$\range(f_\ell)$ is the set $\{f_\ell (q): q \in Q\}$}.} For each $v \in \range(f_\ell)$, we now have to compute the distance between each point $q \in Q$ such that $f_\ell(q) = v$, and each hub $h \in H$ such that $f_\ell(h) = v$ and $h$ was not removed. It might be the case that some points in $Q$ are not located on the same machine as the hubs which are hashed to the same value (and in general, these points might not all fit on one machine, see discussion in \Cref{sec:samp}). However, all that is required in this case is that the hubs can be sent to all machines containing points hashed to the same value: this can be done using prefix sum in a constant number of rounds; since there are at most $10$ hubs for each value in $\range(f_\ell)$, each machine receives at most $10$ hubs. Now machines have the information necessary to locally compute $\close_{ij\ell}(q)$ for all points that they contain, and for each point $q \in Q$ the tuple $(q, \close_{ij\ell}(q))$ is generated.

Finally, observe that line 10 can be implemented in $\Oh(1)$ rounds using sorting and prefix sum.

\section{Sample and Solve}
\label{sec:samp}

In this section, we describe \sas $(Q,p,r)$, which is a subroutine in \ukc and \malg in Sections \ref{sec:uni} and \ref{sec:main}, respectively. \sas $(Q,p,r)$ takes a set $Q$ of at most $n$ points, a sampling parameter $p$, and a radius parameter $r$, it relies on \nns discussed in \Cref{sec:hub}, and produces a set of centers $S \subseteq Q$ such that $\cost(Q,S)=\Oh(r)$.

\begin{algorithm}[h]
\caption{\greedy ($R, h, r$)}\label{algo:greedy}
\KwIn{Set $R$ of at most $n$ points; radius parameter $r \in \R^+$.}
\KwOut{A set $G \subseteq R$ of centers.}

\Begin{
Initialize $G \leftarrow \{h\}$.

\While{$(\exists\, x \in R \text{ such that } d(x,G) > 4c_\rho r)$}{
    Let $x \in R$ be the point farthest from $G$\;
    Add $x$ to $G$\;
}

Return $G$\;
}
\end{algorithm}

\begin{algorithm}[h]
\caption{\sas ($Q, p,r$)}\label{algo:sample}
\KwIn{Set $Q$ of at most $n$ points; probability parameter $p \in (0,1)$; radius parameter $r \in \R^+$.}
\KwOut{{A set $S \subseteq Q$ of centers or report {\sc Fail}.}}
\Begin
	{
	\If{$\left(\size{Q}\leq \lspace=\Oh(n^\delta)\right)$}
	{
	 Call $\mbox{\greedy}(Q,q,r)$ for some arbitrary $q \in Q$, and report the set of centers output by it as $S$.
	}


	Initialize $H=\emptyset$.
    
	 Each point in $Q$ is sampled independently with probability $p$. The points that are selected in this process forms the hub set $H$.
	
	{If $H=\emptyset,$ report {\sc Fail}.}\\
    
    Call  $\mbox{\nns}(Q,H, \rho)$.  If $\mbox{\nns}(Q,H, \rho)$ returns {\sc Fail}, then report {\sc Fail}.\\

    Otherwise, For each point $q$ in $Q$, assign it to the closest hub in $H$ as output by $\mbox{\nns}(Q,H, \rho)$. We call the set of points assigned to a hub $h \in H$ the {\em bag corresponding to $h$}, and denote it as $B_h$. Note that $B_h$ includes $h$.
	
	\For{(each $h \in H$)}
	{
		\If {$(\size{B_h} \leq \lspace)$}
		{
		Collect $B_h$ on a single machine.
		
	$S_h \leftarrow \mbox{\greedy}(B_h,h,r)$.
	}
	\Else{
Form bags $B_{h_1},\ldots,B_{h_w}$, keeping $h$ in every $B_{h_i}$ ($i \in [w]$)  and putting every other point in $B_h \setminus \{h\}$ into exactly one of the $B_{h_i}$'s, such that $\size{B_{h_i}}\leq \lspace=\Oh(n^{\delta})$ for each $i \in [w]$.

	$S_{h_i} \leftarrow \mbox{\greedy}(B_{h_i},h,r)$, where $i \in [w]$.
	
	$S_h \leftarrow \bigcup\limits_{i=1}^w S_{h_i}$.
	
	}
	}
	Report set of centers $S=\bigcup\limits_{h \in H} S_h$.
}
\end{algorithm}

\sas $(Q,p,r)$  calls algorithm \greedy$(R,h,r)$ as a subroutine, which produces a set of centers $G\subseteq R$ such that $\cost(R,G)=\Oh(r)$. \greedy$(R,h,r)$ is a variation of a classic 2-approximation algorithm for $k$-center in the sequential setting~\cite{HS85}. In \sas $(Q,p,r)$,  the idea is to sample each point in $Q$ (independently) with probability $p$ to form a set of hubs $H$. Then each
point $q \in Q$ will be assigned to some hub $h \in H$ by using \nns
(as described in Algorithm~\ref{algo:nn}). For $h \in H$, let $B_h$ be the set of points assigned to $h$ (including $h$ itself). We run \greedy for the points in $B_h$, to produce a set of centers $S_h$.
Finally, $\bigcup_{h \in H} S_h$ is the output reported by \sas. There are other technicalities -- $\size{B_h}$ may be much larger than $\lspace$. In that case, we distribute the points in $B_h \setminus \{h\}$ across a number of machines, but we send $h$ to each machine, ensuring that the total number of points assigned to a machine (including $h$) is less than \lspace---and then we apply \greedy to the points on each of these machines.

The formal algorithm for \sas is presented in Algorithm~\ref{algo:sample}. The approximation guarantee, round complexity and space complexity of \sas are stated in \Cref{lem:samp}. 

\begin{lem}[{\bf Guarantee of \sas}]
\label{lem:samp}
Consider \sas ($Q, p,r$), as described in Algorithm~\ref{algo:sample}. With probability at least $1- e^{- \Omega(p n^\delta) }- \frac{1}{n^{\Omega(1)}}$, it does not report {\sc Fail} and, moreover
\begin{enumerate}
\item[(i)] It produces a set of centers $S\subseteq Q$ such that $\cost(Q,S)\leq 4c_\rho r=\Oh(r)$, where $c_\rho$ is the constant as in \Cref{theo:lhs};
\item[(ii)] It takes $\Oh(1)$ \mpc rounds with local space $\lspace=\Oh\left(n^{\delta} \right)$ and global space $\gspace=\tOh(n^{1+\rho} \cdot \log \Delta)$.
\end{enumerate}
\end{lem}

\begin{rem}
\label{rem:sas}
We call \sas from \ukc (Algorithm~\ref{algo:uniform}) with probability parameter $p =\Omega \left(\frac{\log n}{n^\delta}\right)$. Therefore, the success probability of \sas in our case is always at least $1-\frac{1}{n^{\Omega(1)}}$.
\end{rem}
\begin{proof}[Proof of \Cref{lem:samp}]
Note that \sas (Algorithm \ref{algo:sample}) crucially calls  \greedy (Algorithm~\ref{algo:greedy}) multiple times, particularly in  line numbers {3, 12 and 17}. We start the proof with the following observation (about algorithm \greedy$(R,h,r)$) that follows from the description of Algorithm~\ref{algo:greedy}.

\begin{obs}
\label{obs:grdy}
The output $G \subseteq R$ produced by \greedy($R,h,r$) (as described in Algorithm \ref{algo:greedy}) satisfies $\cost (R,G)\leq 4 c_\rho r$.
\end{obs}

Note that both (i) and (ii) of \Cref{lem:samp} are direct if $\size{Q}\leq \lspace$. As in that case, we executes $\mbox{\greedy}(Q,q,r)$ for some $q \in Q$ in one machine locally, and report its output as $S$. By \Cref{obs:grdy}, we have  $\cost(Q,S)\leq 4c_\rho r$.


{ Now consider the case when $\size{Q} > \lspace$. Note that \sas reports {\sc Fail} if either the set of hubs $H$ is empty, as described in Line~7, or if $\mbox{\nns}(Q, H, \rho)$ reports {\sc Fail}, as described in Line~8. Since each point in $Q$ is added to $H$ independently with probability $p$, the probability that \sas reports {\sc Fail} due to the former event is at most $(1 - p)^{|Q|} \leq e^{-\Omega(p |Q|)} \leq 1-e^{-\Omega(p n^\delta)} $.  Furthermore, since $\mbox{\nns}(Q, H, \rho)$ succeeds with high probability, the overall success probability of \sas is at least $1- e^{- \Omega(p n^\delta)}- \frac{1}{n^{\Omega(1)}}$.}

Now, we argue (i) and (ii) of \Cref{lem:samp} separately.  Recall the description of Algorithm~\ref{algo:sample} from {Line 9--19}.
\begin{enumerate}
\item[(i)] By \Cref{obs:grdy}, for $h \in H$, $\cost(B_h,S_h) \leq 4c_\rho r$. As $S=\bigcup\limits_{h \in H} S_h$ and $Q=\bigcup\limits_{h \in H} B_h$, $\cost(Q,S)\leq 4c_\rho r$.
\item[(ii)] From \Cref{theo:lhs}, \nns can be implemented in \mpc with local space $\lspace= \Oh(n^\delta)$ and global space $\gspace=\tOh(n^{1+\rho} \cdot \log \Delta )$ in $\Oh(1)$ rounds. After \nns is performed, each point knows its assigned hub. Using sorting, we can place all points with the same hubs on consecutive machines in $\Oh(1)$ rounds, and using prefix sum, we can count the number of points assigned to each hub in $\Oh(1)$ rounds. Now, we consider two cases:

    If $|B_h| \leq \lspace $ (that is: the bag could fit on a single machine) then \greedy on $B_h$ can be performed on a single machine without communication, that is, in $0$ rounds.\footnote{A minor technical matter is moving $B_h$ to a single machine if it is big enough to fit but originally stored on two consecutive machines. This can be done in $1$ round if we have $\geq 2\lspace/n$ machines: when sorting, use only the first $\lspace/n$ machines, then if $B_h$ was originally stored on machines $i$ and $i+1$, move it to machine $i + \lspace/n$.}

    If $|B_h| > \lspace$ (that is: the bag could not fit on a single machine) then we arbitrarily partition the bag and perform \sas on each part. Specifically, we send $h$ to each of the consecutive machines on which $B_h$ is stored, and these machines perform \greedy on the subset of the bag that they hold locally. This can be performed in $\Oh(1)$ rounds.
\remove{\begin{cl}\label{cl:bag-bnd}
\complain{ With probability at least $1-\frac{1}{n^{\Omega(1)}}$, $\size{B_h} = \Oh\left(\frac{\log n}{p}\right)$ for each $h \in H$.}
\end{cl}}
\end{enumerate}
\end{proof}

We now state an additional property of \sas in \Cref{lem:samp-add}, which will be useful in both \Cref{sec:uni} and \Cref{sec:main}, where we discuss \ukc and \malg, respectively. Informally, \Cref{lem:samp-add} ensures that within any cluster, once a hub is selected, no non-hub point from that cluster can be chosen as a center.

\begin{lem}[{\bf An additional guarantee of \sas}]
\label{lem:samp-add} {Let $\cC_r$ be a clustering of $Q$ having cost at most $r$.  Then, with high probability, the following holds for any $C \in \cC_r$:} if at least one hub is selected from $C$, then no further point in $C \setminus H$ is selected as a center, that is, $\size{S \cap C} = \size{H \cap C}$.
\remove{\begin{enumerate}
    \item[(i)] if at least one hub is selected from $C$, then no further point in $C \setminus H$ is selected as  a centroid;
    \item[(ii)]  the number of centroids in $C$ does not exceed the number of hubs in $C$, that is, $\size{S \cap C} =\size{H \cap C}$.
\end{enumerate}}
\end{lem}
\begin{proof}
Consider any point $q \in C \setminus H$. As at least one hub is selected from $C$, $d(q,H) \leq 2r$ . By the guarantee from \nns{} (see \Cref{theo:lhs}), with probability at least $1-\frac{1}{n^{\Omega(1)}}$, $q$ is assigned to some hub $h \in H$ such that $d(q,h) \leq 2 c_\rho d(q,H) \leq 4 c_\rho r$. So, when we call \greedy ($B_h, h, r$), as $d(q,h) \leq 4 c_\rho r$, $q$ will not be selected as a center. This implies that $\size{S \cap C} \leq \size{H \cap C}$. The claim follows as $H \subseteq S$.
\end{proof}

\section{Uniform Center algorithm}\label{sec:uni}

In this section, we describe \ukc $(V,r,t)$, {which iteratively refines a set of centers to a smaller set of centers, by calling \sas on a quadratically-increasing probability schedule}. It calls \sas $\Theta(\log \log t)$ times. The $i$-th call to \sas is \sas$(S_{i-1},p_{i-1},r)$ (in particular): it produces a set $S_i\subseteq S_{i-1}$ of centers as the output, where $S_0=V$ and the probability parameters are set suitably. The formal algorithm is described in Algorithm~\ref{algo:uniform}. Its guarantees are stated in \Cref{lem:ukc} --- they follow from the guarantees we have for \sas in \Cref{lem:samp} and the fact that \ukc $(V,r,t)$ calls \sas $\Oh(\log \log t)$ times. \ukc has an additional guarantee as stated in \Cref{lem:uniform3} relating to the reduction of cluster sizes, which plays a crucial role in proving the correctness of \malg in \Cref{sec:main}. In particular it is useful in bounding the number of centers output by \malg.

\begin{algorithm}[h]
\caption{\ukc ($V, r, t$)}\label{algo:uniform}
\KwIn{A set of points $V$ of at most $n$ points, a radius parameter $r \in \R^+$, and an additional parameter $t \leq n$.}
\KwOut{A set $S \subseteq V$ of centers or report {\sc Fail}.}
\Begin{


$p_0=\Theta\left(\frac{\log n}{n^{\delta}}\right)$, $s_0=t$, and $S_0 \leftarrow V$.


\For{ $i=1$ to $\tau=\Theta(\log \log t)$}
{
{Call $\mbox{\sas} \left(S_{i-1},p_{i-1},r\right)$. If it reports {\sc Fail}, then report {\sc Fail}. Otherwise, let $S_{i}$ be its output.}

$s_i \leftarrow \sqrt{s_{i-1}}$ and $p_i= \frac{1}{s_{i}}$.

}
Report $S=S_\tau$.
}
\end{algorithm}

\begin{lem}[{\bf Guarantee of \ukc}] \label{lem:ukc}
Consider  \ukc ($V, r,t$), as described in Algorithm~\ref{algo:uniform}.   With high probability, it does not report {\sc Fail} and, moreover
\begin{enumerate}
    \item[(i)] It produces output $S$ such that
$\mbox{{\sc Cost}}(V ,S) = \Oh(r \cdot \log \log t)$;
\item[(ii)] It takes $\Oh(\log \log t)$ \mpc rounds with local space $\lspace=\Oh\left(n^{\delta} \right)$ and global space $\gspace=\tOh(n^{1+\rho} \cdot \log \Delta)$.
\end{enumerate}
\begin{proof}
    Note that \ukc$(V, r, t)$ makes at most $\Oh(\log \log t)$ calls to \sas. 
In each call \sas$(S_{i-1}, p_{i-1}, r)$, we have 
$p_{i-1} \geq p_0 = \Theta\!\left(\tfrac{\log n}{n^{\delta}}\right)$. 
The success probability follows since each call to \sas succeeds with high probability, 
as guaranteed by \Cref{lem:samp}. 

In the $i$-th iteration of $\mbox{\ukc}(V,r,t)$, we call $\mbox{\sas}\left(S_{i-1},p_{i-1}, r \right)$, and get $S_i$ as the centers. By \Cref{lem:samp} (i), $\cost(S_{i-1},S_i)=\Oh(r)$, where $1\leq i \leq \tau$. Hence,
\begin{align*}
    \cost(V,S) &\leq 
    \sum_{i=1}^{\tau}\cost (S_{i-1},S_i) = 
    \tau \cdot \Oh(r)=    \Oh(r \log \log t)
    \enspace.
\end{align*}
So, we are done with the proof of \Cref{lem:ukc} (i). 
Furthermore, by \Cref{lem:samp}(ii) and using the fact that \ukc calls \sas $\Oh(\log \log t)$ times,   \Cref{lem:ukc}(ii) follows.

\end{proof}
\end{lem}

\begin{lem}[{\bf Reduction in cluster sizes}]\label{lem:uniform3}
Consider \ukc ($V, r,t$) as described in Algorithm~\ref{algo:uniform}, and a fixed clustering $\cC_r^t$ of $V$ that has cost $r$. It produces output $S\subseteq V$ such that the following holds for any $C \in \cC_r^t$: if  $\size{ C \cap V}\leq t$, then with probability at least $1-\frac{1}{t^{\Omega(1)}}$, we have  $\size{C \cap S} =\Oh\left( \log t \cdot (\log \log t)^2\right)$.

\end{lem}

\begin{proof}
Let  $b_i=(1+\eta)^is_i \log t\cdot \left( \log \log t \right)^2$, where $i$ is an non-negative integer and $\eta=\Theta\left(\frac{1}{\log \log t}\right)$. Observe that

$$ b _{\tau-1} =\left(1+\eta\right)^{\tau-1} s_{\tau-1} \log t (\log \log t)^2=\Oh\left(\log t \cdot \left( \log \log t\right)^2\right).$$ Using induction on $i$ ($i \in \N$), we will show that $\size{C \cap S_i}\leq b_{i-1}$ for each $i$ with $1 \leq i \leq \tau$, with probability at least $1-\frac{1}{t^{\Omega(1)}}$.  This will imply the desired result as $S_i$ is the output after the $i$-th iteration, and  $S=S_\tau$. Hence, $\size{C \cap S } \leq b_{\tau -1} =\Oh\left(\log t \cdot (\log \log t)^2\right)$.

 For $i=1$,
 $$\size{C \cap S_1}\leq \size{C \cap S_0} = \size{C \cap V}\leq t \leq b_{0}.$$
 The first inequality follows as $S_1 \subseteq S_0$; the second equality follows as $S_0=V$; the third inequality follows from the given condition that $\size{C \cap V}\leq t$; and the fourth one holds by the definition of $b_0$.

\remove{ $b_{-1}=b_0$, $\size{C \cap V} \leq t$, $S_0=V$, and by the definition of $b_0$.
 we first show that $\size{C \cap S_1} \leq b_0$ holds with probability at least $1-\frac{1}{t^{\Omega(1)}}$. In iteration $1$ of \ukc,  it calls $\mbox{\sas}(S_{0},\frac{1}{s_{0}},r)$, and produces $S_1$ as the intermediate centroids. As $\size{C \cap S_0}\leq t$, the expected size of $C \cap S_1$ is at most $\frac{t}{s_0}$. So, by using Chernoff bound, we have that $\size{C \cap S_1} \leq \frac{t \log t}{s_0}=(t \log t)^{\eps} \leq b_0$, with probability at least $1-\frac{1}{t^{\Omega(1)}}.$}

 Suppose the statement holds for each $i$ with $1 \leq i \leq \ell-1$, that is, $\size{C \cap S_i}\leq b_{i-1}$ for each $i$ with $1\leq i \leq \ell-1$. Now we argue for $i=\ell$. If $\size{C \cap S_{\ell-1}} \leq b_{\ell -1}$, then $\size{C \cap S_{\ell}}\leq b_{\ell -1}$, as  $S_\ell \subseteq S_{\ell -1}$. So, let us assume that $\size{C \cap S_{\ell-1}} > b_{\ell-1}$.

Consider the $\ell$-th iteration of \ukc: it calls algorithm $\mbox{\sas}(S_{\ell-1},p_{\ell -1},r)$, and produces $S_\ell$ as the set of intermediate centers. Let $H_\ell \subseteq S_{\ell-1}$ be the set of hubs sampled in the call of $\mbox{\sas}(S_{\ell-1},p_{\ell -1},r)$, where each point in $S_{\ell-1}$ (independently) included in $H_\ell$ with probability $p_{\ell -1}$.

 Before proceeding, we make two observations:
\begin{obs}\label{obs:obs1} The probability, that at least one point from $C \cap S_{\ell-1}$ is in $H_\ell$, is at least $1-\frac{1}{t^{\Omega(1)}}.$
\end{obs}
\begin{proof}
The probability, that no point in $C \cap S_{\ell-1}$ ($\size{C \cap S_{\ell-1}} > b_{\ell-1}$) is included in $H_\ell$, is at most
$$\left(1-p_{\ell -1}\right)^{b_{\ell -1}} = \left(1-\frac{1}{s_{\ell -1}}\right)^{b_{\ell -1}}\leq  \frac{1}{t^{\Omega(1)}}.$$ Here, we have used that $b_{\ell -1} =(1+\eta)^{\ell - 1}s_{\ell - 1}\log t \cdot (\log \log t)^2$.
\end{proof}

 \begin{obs}\label{obs:obs2}
With probability at least $1-\frac{1}{t^{\Omega(1)}}$, the number of points in $C \cap S_{\ell -1}$ that are in $H_\ell$ is at most  $b_{\ell-1}$.
 \end{obs}
 \begin{proof}
 By induction hypothesis, $\size{C \cap S_{\ell-1}}\leq b_{\ell-2}$. The expected number of points of $C \cap S_{\ell-1}$ that are in $H_\ell$ is
 \begin{eqnarray*}
 {\size{C \cap S_{\ell-1}}}\cdot \frac{1}{s_{\ell-1}} &\leq& \frac{b_{\ell-2}}{s_{\ell-1}}\\ &=& \frac{(1+\eta)^{\ell-2}s_{\ell-2}\log t \cdot (\log \log t)^2}{s_{\ell-1}}\\
 &=&(1+\eta)^{\ell-2}s_{\ell -1} \log t \cdot (\log \log t)^2=\mu.
 \end{eqnarray*}

By using Chernoff bound (\Cref{lem:cher}),  the probability, that  the number of points of $C \cap S_{\ell-1}$ that are in $H_\ell$ is more than $(1+\eta) \mu=b_{\ell-1}$, is at most $e^{-\frac{\eta^2\mu}{3}} \leq \frac{1}{t^{\Omega(1)}}.$
 \end{proof}

Observations \ref{obs:obs1} and \ref{obs:obs2}, along with \Cref{lem:samp-add}, give us that $\size{C \cap S_\ell} = \size{C \cap H_\ell} \leq b_{\ell -1}$ with probability at least $1-\frac{1}{t^{\Omega(1)}}$.
\end{proof}

\section{The main algorithm \remove{(Proof of \Cref{theo:main1})}}
\label{sec:main}

In this section, we present our main algorithm \malg. Recall the overall description of \malg in \Cref{sec:over}. \malg has two phases. In {\bf Phase 1}, it calls \ukc $\alpha$ times, and in {\bf Phase 2}, it calls \sas $\beta$ times, where $\alpha$ is the input precision parameter and  $\beta=\Theta(\log ^{(\alpha +1)} n)$. The formal algorithm is described in Algorithm~\ref{algo:main}. We establish its approximation guarantee, round complexity, and global space in \Cref{lem:main}, and we bound the number of centers in \Cref{lem:main3}. Combining \Cref{lem:main} and \Cref{lem:main3}, we have \Cref{theo:main1}.

\begin{algorithm}[h]
\caption{\malg ($P, r$)}\label{algo:main}
\KwIn{Set $P$ of $n$ points; tradeoff parameter $\alpha$; radius parameter $r \in \R^+$.}
\KwOut{A set $T\subseteq P$ of centers or report {\sc Fail}.}
\Begin{


{\bf Phase 1:}

$T_0 \leftarrow P$, $t_0=n $, and $r_0= \frac{r}{\log \log t_0}=\frac{r}{\log \log n}$.

\For{$(j=1~\mbox{to}~\alpha)$}
{
 {\bf Phase 1.j:}
 
Call $\mbox{\ukc} (T_{j-1},r_{j-1}, t_{j-1})$. If it reports {\sc Fail}, then report {\sc Fail.} Otherwise, 
$T_{j}$ be its  output. 

$t_{j}=\Theta(\log t_{j-1} \cdot (\log \log t_{j-1})^{d+2})$. Note that $t_j=\widetilde{\Theta}(\log t_{j-1})=\widetilde{\Theta}(\log ^{(j)} n).$ 


$r_j=\frac{r}{\log \log  t_{j}}.$

}
{\bf Phase 2:}




\For{ $(i=1~\mbox{to}~\beta= \Theta (\log ^{(\alpha +1)} n))$}
{

{\bf Phase 2.i:}

Call $\mbox{\sas} (T_{\alpha+i-1},\frac{1}{2},r)$. If it reports {\sc Fail}, then report {\sc Fail}. Otherwise, let $T_{\alpha+i}$ be its output.


}

Report $T=T_{\alpha+\beta}$.
}
\end{algorithm}

\begin{lem}[{\bf Guarantee of \malg}]\label{lem:main}
    Consider \malg ($P,t$), as described in Algorithm~\ref{algo:main}. It does not report {\sc Fail} with high probability, and moreover:
    \begin{enumerate}
        \item[(i)] It produces output $T \subseteq P$ such that $\mbox{{\sc Cost}}(P ,T) = \Oh(r \cdot(\alpha+ \log ^{(\alpha +1)} n))$;
        \item[(ii)] It takes $\Oh(\log \log n)$ rounds of \mpc with $\lspace=\Oh(n^\delta)$ and global space $\gspace=\tOh(n^{1+\rho} \cdot \log \Delta)$.
    \end{enumerate}
\end{lem}
\begin{proof}
    The success probability follows from the success probability guarantees of \ukc and \sas (see \ref{lem:ukc} and \Cref{lem:samp}, respectively).

For \Cref{lem:main} (i),
observe that
\begin{align*}
    \cost(P,T) &= \cost(T_0,T_{\alpha+\beta}) 
    \leq \cost(T_0,T_\alpha)+\cost(T_\alpha,T_{\alpha+\beta})
    \enspace.
\end{align*}
It therefore suffices to show that $\cost(T_0,T_\alpha)$ and $\cost(T_\alpha,T_{\alpha+\beta})$ are bounded by $\Oh(r \alpha)$ and $\Oh(r \cdot \log^{(\alpha+1)} n)$, respectively.

For any $j$ with $1\leq j \leq \alpha$, note that \malg($P, r$) calls $\mbox{\ukc}(T_{j-1},r_{j-1},t_{j-1})$ in {\bf Phase 1.j} and produces $T_j$ as the output. So, by \Cref{lem:ukc} (i), $\cost (T_{j-1},T_j) = \Oh(r_{j-1} \cdot \log \log t_{j-1})$, which is $\Oh(r)$. Hence,
\begin{align*}
    \cost(T_0,T_\alpha) &\leq 
    \sum_{j=1}^\alpha \cost \left(T_{j-1},T_j \right) = 
    \alpha \cdot \Oh(r) =
    \Oh\left(r \cdot \alpha\right)
    \enspace.
\end{align*}

For any $i$ with $1\leq i \leq \beta$, note that  \malg($P, r$) calls $\mbox{\sas}(T_{\alpha+i-1},1/2,r)$ in {\bf Phase 2.i} and produces $T_{\alpha +i}$ as the output. So, by \Cref{lem:samp} (i), $\cost (T_{\alpha+i-1},T_{\alpha+i}) = \Oh(r).$ Hence, as $\beta=\Theta(\log ^{(\alpha +1)} n)$,
\begin{align*}
    \cost\left(T_\alpha,T_{\alpha +\beta}\right) &\leq 
    \sum\limits_{i=1}^\beta \cost \left(T_{\alpha+i-1},T_{\alpha + i} \right)
        =
    \Oh(r \cdot \log ^{(\alpha +1)} n)
    \enspace.
\end{align*}
So, we are done with the proof of \Cref{lem:main} (i). Now we prove \Cref{lem:main} (ii).

For any $j$ with $1\leq j \leq \alpha$, in {\bf Phase 1.j},  \malg ($P, r$) calls $\mbox{\ukc}(T_{j-1},r_{j-1},t_{j-1})$. By Lemma \ref{lem:ukc} (ii), the total number of rounds spent by \malg ($P, r$) in {\bf Phase 1} is $\sum_{j=1}^\alpha \Oh\left(\log \log t_{j-1}\right)$ $=\Oh\left(\log \log n\right)$. This is because $t_0=n$, $t_j=\widetilde{\Theta}\left(\log t_{j-1}\right)$ for any $j\geq 1$, that is, $t_j=\widetilde{\Theta}\left(\log ^{(j)} n\right)$\remove{(Minor thing: I assume this should be $\widetilde{\Theta}\left(\log^{(j)} n \right)$: with the tilde over the theta)}. In {\bf Phase 2}, \malg ($P, r$) calls \sas for  $\beta=\Theta\left(\log ^{(\alpha +1)} n\right)$ times. By \Cref{lem:samp} (ii), the total number of rounds spent by \malg ($P, r$) in {\bf Phase 2} is $\Oh(\beta)=\Oh\left(\log ^{(\alpha +1)} n\right)$. So, the round complexity of \malg follows. The global space complexity of \malg follows from the global space complexities of \ukc and \sas (see Lemma  \ref{lem:ukc} (ii) and \Cref{lem:samp}(ii), respectively).

\end{proof}

\begin{lem}[{\bf Number of centers reported by \malg}]
\label{lem:main3}
Consider \malg $(P, r)$ as described in Algorithm~\ref{algo:main}.  It produces output $T$ such that, with probability at least $1-\frac{1}{(\log ^{(\alpha-1)} n)^{\Omega(1)}}$,
\begin{align*}
    \size{T} &\leq 
    \size{\cC_r}\left(1+\frac{1}{\widetilde{\Theta}(\log ^{(\alpha)} n)}\right)+\widetilde{\Theta}((\log ^{(\alpha)} n)^3) + {\widetilde{\Oh}(\log n)}
    \enspace.
\end{align*}
Here, $\cC_r $ is a clustering of $P$ that has the minimum number of centers among all possible clustering of $P$ with cost at most $r$.
\end{lem}

Now, we introduce the notion of \emph{active} and \emph{inactive} clusters in the following definition, which is useful in proving \Cref{lem:main3}. {Inactive clusters are clusters which, at some point during {\bf Phase 1}, fail to reduce in size sufficiently. After the sub-phase during which they fail to reduce in size sufficiently, we assume that they never reduce in size again (since this is the worst case). We are then able to bound the total number of centers in inactive clusters (\Cref{lem:inter1}). Active clusters, by contrast, always reduce in size as we expect: the number of centers in active clusters is therefore easy to bound.}

\begin{defi}
\label{defi:active}
Let $\cC_r$ be an optimal clustering with cost at most~$r$. For each $C \in \cC_r$ and $j$ with $1 \leq j \leq \alpha$, we say $C$ is \emph{inactive} in \textbf{Phase 1.j} if $\size{C \cap T_i} > t_i$ for some $i$ with $1 \leq i < j$. Otherwise, if $\size{C \cap T_i} \leq  t_i$ for every $i$ with $1 \leq i < j$, $C$ is called \emph{active} in~\textbf{Phase~1.j}. 
\end{defi}
 
Let $\cC_r' \subseteq \cC_r$ be the set of clusters that are active after {\bf Phase 1}, that is, {\bf Phase 1.$\alpha$}. By the definition of active clusters, for each $C \in \cC_r'$, $\size{C \cap T_\alpha} \leq t_{\alpha}$. Note that \malg goes over $\beta$ sub-phases in {\bf Phase 2}. After {\bf Phase 1} and before the start of {\bf Phase 2}, it has $T_{\alpha}$ as the set of intermediate centers. For $1 \leq i \leq \beta$, in {\bf Phase 2.i}, we call $\mbox{\sas}\left(T_{\alpha+i-1},\frac{1}{2},r \right)$, and get $T_{\alpha+i}$ as the intermediate centers. For $0 \leq i \leq \beta$; a cluster $C \in \cC_r'$ is said to be $i$-\emph{large} if $\size{C \cap T_{\alpha+i}} \geq 2$. Let $\Gamma_i \subseteq \cC_r'$ denote the set of $i$-\emph{large} clusters, and let $Y_i$ denote the total number of points that are {in $i$-large clusters}, that is, $Y_i=\sum\limits_{C \in \Gamma_i} \size{C \cap T_{\alpha+i}}$.

Note that, in \Cref{lem:main3}, we want to bound  the number of centers in $T=T_{\alpha+\beta}$. We first observe that $\size{T}$ can be expressed as the sum of three quantities:

\begin{obs}
\label{obs:boundT}
$\size{T}=\size{T_{\alpha+\beta}}\leq \size{\cC_r}+Y_{\beta}+\sum\limits_{C \in \cC_r \setminus \cC_r'} \size{C \cap T_\alpha}$.
\end{obs}

\begin{proof}
Observe that 
since $T_{\alpha+\beta} \subseteq T_\alpha$, we obtain,
\begin{align*}
    T_{\alpha +\beta} &=
    \sum_{C \in \cC_r \setminus \cC_r'}\size{C \cap T_{\alpha+\beta}}+\sum_{C \in \cC_r'}\size{C \cap T_{\alpha+\beta}}
    \\&\leq
    \sum_{C \in \cC_r \setminus \cC_r'}\size{C \cap T_{\alpha}}+\sum_{C \in \cC_r'}\size{C \cap T_{\alpha+\beta}}
    \enspace.
\end{align*}
To bound the second term by $\size{\cC_r}+Y_\beta$, observe that
\begin{align*}
    \sum_{C \in \cC_r'}\size{C \cap T_{\alpha+\beta}}
        &=
    \!\!\!\sum_{C \in \cC_r' : \size{C \cap T_{\alpha + \beta}}=1 }\!\!\!\size{C \cap T_{\alpha+\beta}} +
    \!\!\!\!\!\!\sum_{C \in \cC_r':\size{C \cap T_{\alpha + \beta}}\geq 2}\!\!\!\size{C \cap T_{\alpha+\beta}}
        \\&\leq
    \size{\cC_r'}+Y_\beta \leq \size{\cC_r}+Y_\beta 
    \enspace,
\end{align*}
which used that $\cC_r'\subseteq \cC_r$.
This yields \Cref{obs:boundT}.
\end{proof}

In the following lemmas, we bound \remove{$\size{T}$ (as claimed in \Cref{lem:main3}) by bounding} $\sum_{C \in \cC_r \setminus \cC_r'} \size{C \cap T_\alpha}$ and $Y_\beta$, and (with \Cref{obs:boundT}) the result of \Cref{lem:main3} immediately follows from these bounds.  Lemmas~\ref{lem:inter1}~and~\ref{lem:inter2} 
 are technical that we will prove later.
 
\begin{lem}
\label{lem:inter1}
With probability at least $1-\sum_{i=1}^{\alpha}\frac{1}{t_{i-1}^{\Omega(1)}}$, $\sum_{C \in \cC_r \setminus \cC_r'} \size{C \cap T_\alpha}$ is \edit{$\Oh\left( \frac{ \size{\cC_r}}{(\log ^{(\alpha  ) } n)^{\Omega(1)}}\right)$}, that is, the number of points in $T_{\alpha}$ that are  present in clusters that are inactive after {\bf Phase 1} is   ${\Oh\left( \frac{ \size{\cC_r}}{(\log ^{(\alpha)} n)^{\Omega(1)}}\right)} + {\widetilde{\Oh}(\log n)}.$
\end{lem}

\begin{lem}
\label{lem:inter2}
With probability at least $1-\frac{1}{t_{\alpha - 1}^{\Omega(1)}}$, we have $Y_\beta = \Oh\left(\frac{\size{\cC_r}}{t_\alpha}+t_\alpha^3 \cdot \log t_\alpha\right)$.
\end{lem}

\begin{proof}[Proof of \Cref{lem:main3} using \Cref{lem:inter1} and \Cref{lem:inter2}] From the above two lemmas along with \Cref{obs:boundT} and the fact $t_j=\widetilde{\Theta}(\log ^{(j)} n)$, we have the following bound on $\size{T}$ with probability at least $1-\sum\limits_{i=1}^{\alpha} \frac{1}{t_{i-1}^{\Omega(1)}}\geq 1-\frac{1}{(\log ^{(\alpha - 1)} n)^{\Omega(1)}}$:
\begin{align*}
    \size{T} \leq \size{\cC_r}\left(1+\frac{1}{\widetilde{\Theta}(\log ^{(\alpha)} n)}\right)+\widetilde{\Theta}\left((\log ^{(\alpha)} n)^3\right) + {\widetilde{\Oh}(\log n)}
    \enspace.
\end{align*}
Hence, we are done with the proof of \Cref{lem:main3}.
\end{proof}

\subsection*{Proof of \Cref{lem:inter1}}

We prove \Cref{lem:inter1} by using the following lemma, which we prove later.

\begin{lem}
\label{lem:main-P1i}
Consider \malg ($P, r$) as described in Algorithm \ref{algo:main} and {\bf Phase 1}. With probability at least $1-\sum_{i=1}^{\alpha}\frac{1}{t_{i-1}^{\Omega(1)}}$, the following hold: 
\begin{enumerate}
    \item[(i)] there is no cluster $C \in {\cC_r}$ that is inactive after {\bf Phase 1.1}; 
    \item[(ii)] for $i$ with  $2\leq i \leq \alpha$, the number of clusters $C \in {\cC_r}$ such that $C$ is active in {\bf Phase 1.i} but inactive after {\bf Phase 1.i}, is at most \edit{$\Oh \left(\frac{\size{\cC_r}}{t_{i-1}^{\Omega(1)}}+{\log t_{i-1}}\right)$}.
\end{enumerate}

\end{lem}

\begin{proof}[Proof of \Cref{lem:inter1} using \Cref{lem:main-P1i}]
Let us partition the clusters in {$\cC_r \setminus \cC_r'$} into $\mathcal{A}_1,\ldots,\mathcal{A}_{\alpha}$, where $\mathcal{A}_i$ is the set of clusters that are active in {\bf Phase 1.i} but inactive after {\bf Phase 1.i}, where $1 \leq i \leq \alpha$. Consider a cluster $C \in \mathcal{A}_i$. By the definition of $\mathcal{A}_i$, $\size{C \cap T_{i-1}}\leq t_{i-1}$. That is, for each $C \in \mathcal{A}_i$, $\size{C \cap T_\alpha} \leq t_{i-1}$, because $T_0 \supseteq T_1 \supseteq \ldots \supseteq T_{\alpha}$. 

Applying \Cref{lem:main-P1i} , we have $\size{\cA_i}=0$ and $\size{\mathcal{A}_i}=\edit{\Oh\left(\frac{\size{\cC_r}}{t_{i-1}^{\Omega(1)}}+ {\log t_{i-1}}\right)}$ for each $i~(2 \leq i \leq \alpha)$, with probability at least $1-\sum_{i=1}^{\alpha}\frac{1}{t_{i-1}^{\Omega(1)}}$. Hence, with the same probability,
\begin{align*}
    \sum_{C \in \cC_r \setminus \cC_r'} \size{C \cap T_\alpha} &\leq
    \sum_{i=1}^{\alpha} \sum_{C \in \mathcal{A}_i}\size{T_\alpha \cap C} \leq 
    \sum_{i=2}^{\alpha} \size{\mathcal{A}_i} t_{i-1}
        \\ &=
    \Oh\left(\sum_{i=2}^\alpha \left(\frac{\size{\cC_r}}{t_{i-1}^{\Omega(1)}}+ \log t_{i-1}\right) \cdot t_{i-1}\right) =
    \Oh\left( \frac{ \size{\cC_r}}{\left(\log ^{(\alpha  )}n\right)^{\Omega(1)}}\right) + \widetilde{\Oh}(\log n)
    \enspace.
\end{align*}
The last step uses that $t_i=\widetilde{\Theta}\left(\log ^{(i)} n\right)$.
%
%
\end{proof}

\begin{proof}[Proof of \Cref{lem:main-P1i}]
\remove{We prove the lemma by using induction on $j$. The claim is trivially true for $j=0$. Assume that the lemma is true for every $j$ with $0 \leq j \leq \ell -1$. Now we consider the case when $j=\ell$. }

 Consider $i$ with $1 \leq i \leq \alpha$. Let $\mathcal{A}_i$ be the set of clusters that were active in {\bf Phase 1.i} but inactive after {\bf Phase 1.i}, and let $\mathcal{B}_i \supseteq \mathcal{A}_i$ be the set of clusters that were active in {\bf Phase 1.i}. It suffices to show that $\lvert \cA_1 \rvert = 0$ with probability at least $1 - \frac{1}{n^{\Omega(1)}} \geq 1 - \frac{1}{t_0^{\Omega(1)}}$, and that for $i \ge 2$, 
$\lvert \mathcal{A}_i \rvert = \Oh\!\Big(\frac{\lvert \cC_r \rvert}{t_{i-1}^{\Omega(1)}} + \log t_{i-1}\Big)$ holds with probability at least $1 - \frac{1}{t_{i-1}^{\Omega(1)}}$.

 \remove{By the induction hypothesis, with probability at least $1-\sum_{i=0}^{\max\{0,\ell-2\}}\frac{1}{t_{i-1}^{\Omega(1)}}$, $\size{\mathcal{A}_i}= o\left(\frac{\size{\cC_r}}{t_{i-1}^{\Omega(1)}}\right)$ for each $i$ with $0 \leq i \leq \ell-1$.}
For $C \in \mathcal{B}_{i}$, let $X_C$ be the random variable defined as
\begin{align*}
    X_C &= \begin{cases}
        1 & \text{ if } \size{C \cap T_{i}} > t_i \\
        0 &  \mbox{otherwise.}
    \end{cases}
\end{align*}

Observe that $\size{\mathcal{A}_i} =\sum_{C \in \mathcal{B}_i} X_C$.

\begin{cl}
\label{cl:X_C}
The probability that $X_C=1$ is $\Oh\left(\frac{1}{t_{i-1}^{\Omega(1)}}\right).$
\end{cl}

Using the above claim, in order to prove \Cref{lem:main-P1i}, let us consider the case $i=1$ and that of $i \geq 1$, separately.

For $i=1$, applying the union bound over all $C \in \mathcal{B}_i=\mathcal{B}_1$, the probability that there exists a $C \in \mathcal{B}_1$ such that $X_C=1$ is at most $\frac{\size{\mathcal{B}_1}}{t_0^{\Omega(1)}}<\frac{1}{t_0^{\Omega(1)}}$. It is because $\size{\mathcal{B}_1}\leq n$ and $t_0=n$. This implies that $\size{\mathcal{A}_1}=0$ with probability at least $1-\frac{1}{t_0^{\Omega(1)}}$. Note that we are done for the case $i=1.$

For $2 \leq i \leq \alpha$, $\E [\size{\mathcal{A}_i}] =\Oh\left(\frac{\size{\mathcal{B}_i}}{t_{i-1}^{\Omega(1)}}\right)=\Oh\left(\frac{\size{\cC_r}}{t_{i-1}^{\Omega(1)}}\right)$.  Using  Chernoff bound (\Cref{lem:cher}), 
\begin{align*}
  \Pr \left(\size{\mathcal{A}_i}  \geq c_1 \cdot  \left(\frac{\size{\cC_r}}{t_{i-1}^{\Omega(1)}} + \log t_{i-1}\right)\right)
    \leq \frac{1}{t_{i -1 }^{\Omega(1)}},
    \end{align*}
where $c_1$ is a suitable large constant. Thus, by applying the union bound over all $i$ with $1 \leq i \leq \alpha$, the proof of \Cref{lem:main-P1i} is complete, and it only remains to prove \Cref{cl:X_C}.

\begin{proof}[Proof of \Cref{cl:X_C}]
%
{Consider the clustering $\cC_{r_{i-1}}^{t_{i-1}}$ of $T_{i-1}$ as follows. For each $C \in \cC_r$, consider the partition of cluster $C$ into  $z$ many clusters $C_1.\ldots, C_z$ such that the radius of each $C_i$ is at most $r_{i-1}$. To bound $z$, observe that any cluster $C \in \cC_r$ lies within a ball of radius at most $r = r_{i-1} \cdot \log \log t_{i-1}$. We partition this ball using a regular grid of side length $r_{i-1}$, which ensures that each grid cell lies within a ball of radius $O(r_{i-1})$. The number of such cells intersecting the ball is at most the ratio of volumes, i.e.,
$
z \leq  O( ( {r}/{r_{i-1}} )^d ) = O\left( (\log \log t_{i-1})^d \right).
$
For each $C_i$ ($i \in [z]$), the corresponding cluster in $\cC_{r_{i -1}}^{t_{i -1}}$ is $C_i'=C_i \cap T_{i -1}$. So, $\size{\cC_{r_{i -1}}^{t_{i-1}}}=\size{\cC_r} \cdot \Oh \left(\left( \log \log t_{i -1} \right)^d\right).$}

Consider the particular $C \in \mathcal{B}_{i}$. By \Cref{defi:active}, $\size{C \cap T_{i-1}} \leq t_{i-1}$. If we consider the partition of $C$ into $C_1,\ldots,C_z$ in $\cC_{r_{i-1}}^{t_{i -1}}$, then $\size{C_y \cap T_{i -1}} \leq t_{i-1}$ for each $y \in [z]$. Let $\mathcal{B}_i'$ be the set of clusters in $\cC_{r_{i-1}}^{t_{i -1}}$ that are formed due to the partition of some $C \in \mathcal{B}_i$. So, for each $C' \in \mathcal{B}_i'$, we have $\size{C' \cap T_{i -1}}\leq t_{i-1}.$

Consider {\bf Phase 1.i} of \malg: we get  $T_i$ as the current set of centers by calling the procedure $\mbox{\ukc}(T_{i-1}, r_{i-1}, t_{i-1})$. Let us apply \Cref{lem:uniform3} with $t=t_{i-1}$, $V_t=T_{i-1}$, $r=r_{i-1}$, $S=T_i$, and $\cC_{r}^t=\cC_{r_{k-1}}^{t_{k-1}}$. For $y\in [z]$, with probability at least $1-\frac{1}{t_{i-1}^{\Omega(1)}}$, we have $\size{C_y \cap T_i} =\Oh \left(\log t_{i -1}\cdot \left(\log \log t_{i -1}\right)^2\right)$. Applying union bound for all $y \in [z]$, with probability at least $1-\frac{1}{t_{i-1}^{\Omega(1)}}$, we have $$\size{C \cap T_i} =\sum_{y \in [z]}\size{C_y \cap T_i} =\Oh \left( \log t_{i-1} \cdot \left(\log \log t_{i-1}\right)^{d+2}\right)=t_i.$$ This is because $z=\Oh\left((\log \log t_{i-1})^{d+2}\right)$, and we  are done with the proof of the claim.
\remove{ the  number of clusters $C' \in \cC_{r_{\ell-1}}^{t_{\ell-1}}$'s, such that $\size{ C' \cap T_{\ell-1}}\leq t_{\ell-1}$ but $\size{C' \cap T_\ell} \geq \log t (\log \log t)^2$, is  $o\left(\frac{\size{C_{r_{\ell-1}}^{t_{\ell-1}}}}{ t_{\ell-1}^{\Omega(1)}}\right)$, which is $o\left(\frac{\size{C_{r}}}{ t_{\ell-1}^{\Omega(1)}}\right)$. It is because $\size{\cC_{r_{\ell -1}}^{t_{i-1}}}=\size{\cC_r} \cdot \Oh \left(\left( \log \log t_{\ell -1} \right)^d\right).$

This

This implies $\size{\mathcal{A}_\ell'} =o\left(\frac{\size{C_{r}}}{ t_{\ell-1}^{\Omega(1)}}\right).$

Consider a cluster $C \in \cC_r$ that was active after {\bf Phase 1.(k-1)}, that is, $\size{C \cap T_{k-1}} \leq t_{k-1}$.
}
\end{proof}
\begin{rem}
    In the above proof, while bounding $z$, note that we crucially use the fact that the input is from the $d$-dimensional Euclidean space.
\end{rem}
This completes the proof of \Cref{lem:inter1}.
\end{proof}

\subsection*{Proof of \Cref{lem:inter2}}

We prove \Cref{lem:inter2} by using the following lemma, which we prove later.

\begin{lem}
\label{lem:Y_i}
Let $\zeta \in (0,1)$ be a suitable constant, let $i$ be such that $1\leq i \leq \beta$ with $Y_{i-1} =\Omega\left(t_{\alpha}^3 \log t_{\alpha}\right)$. With probability $1-\frac{1}{t_{\alpha}^{\Omega(1)}}$,
$Y_{i} \leq \sum_{C \in \Gamma_{i-1}} \size{C \cap T_{\alpha + i}} \leq \zeta \cdot Y_{i-1}.$
\end{lem}

Applying the union bound over all $i$'s in $1$ to $\beta$, with probability at least $1-\frac{1}{t_{\alpha}^{\Omega(1)}}$, we have $Y_\beta \leq \zeta^{\beta}Y_0+t_\alpha^3 \log t_\alpha.$ As  $Y_0=\Oh\left(\size{\cC_r'} \cdot t_\alpha \right)$ $=\Oh\left(\size{\cC_r} \cdot t _\alpha\right)$, $\beta=\Theta\left(\log ^{(\alpha + 1)} n\right)$ and $\zeta$ is chosen suitably, we are done with the proof of \Cref{lem:inter2}.
\qed

\begin{proof}[Proof of \Cref{lem:Y_i}]
As $Y_i=\sum_{C \in \Gamma_i} \size{C \cap T_{\alpha + i}}$ and $\Gamma_i \subseteq \Gamma_{i-1}$, $Y_{i} \leq \sum_{C \in \Gamma_{i-1}} \size{C \cap T_{\alpha + i}} $ follows.

For the other part, $\sum_{C \in \Gamma_{i-1}} \size{C \cap T_{\alpha + i}} \leq \zeta \cdot Y_{i-1}$, let $Z_C=\size{C \cap T_{\alpha + i}}$ and $Z=\sum_{C \in \Gamma_{i-1}} Z_C$. Now consider the following claim, which we will prove right after proving \Cref{lem:Y_i}.

\begin{cl}\label{cl:const-P2}
Let $C \in \Gamma_{i-1}$. The probability that $Z_C=\size{C \cap T_{\alpha + i}} \leq \frac{3}{4} \cdot \size{C \cap T_{\alpha +i-1}}$ is at least a constant $\kappa \in (0,1)$.
\end{cl}
Note that $T_{\alpha+i} \subseteq T_{\alpha + i -1}$. So, $\size{C \cap T_{\alpha+i}} \leq \size{C \cap T_{\alpha+i-1}}$ always. For $C \in \Gamma_{i-1}$,  $\size{T_{\alpha + i -1}} \leq t_{\alpha}$.
From the above claim
\begin{align*}
    \E[Z_C] &=
    \E[\size{C \cap T_{\alpha+i}}]
        \\&\leq
    \kappa \cdot \frac{3}{4} \size{C \cap T_{\alpha+i-1}}+ (1-\kappa) \cdot \size{C \cap T_{\alpha +i -1 }}
        \\&\leq
    \zeta^{'} \size{C \cap T_{\alpha +i -1 }}
    \enspace.
\end{align*}

where $\zeta'$ is a suitable constant. Recalling the definition of $Y_{i-1}$, we have $$\E[Z]\leq \zeta^{'}\sum_{C \in \Gamma_{i-1}}\size{C \cap T_{\alpha+i-1}}=\zeta' Y_{i-1}.$$

Moreover, $0 \leq Z_C \leq t_{\alpha}$ for each $C \in \Gamma_{i-1}$. Hence, applying Hoeffding bound (\Cref{lem:hoeff}), we have
\begin{align*}
    \pr\left(Z \geq \zeta Y_{i-1}\right)
        &= 
    \pr\left( Z \geq \E[Z] + \zeta_1 Y_{i-1} \right) ~~~~~~~~(\mbox{where $\zeta=\zeta'+\zeta_1$}.)
        \\&\leq 
    e^{-\frac{\zeta_1^2 Y_{i-1}^2}{\size{\Gamma_{i-1}}t_{\alpha}^2}} \leq e^{-\frac{\zeta_1^2 Y_{i-1}^2}{\size{\Gamma_{i-1}}t_{\alpha}^2}}
        \leq 
    {1}/{t_{\alpha - 1}^{\Omega(1)}}
    \enspace.
\end{align*}
The last inequality folllows from $Y_{i-1}\geq 2 \size{\Gamma_{i-1}}~\mbox{and}~Y_{i-1}=\Omega(t_{\alpha}^3 \cdot \log t_{\alpha})$. This concludes the proof of \Cref{lem:Y_i} since we have $Z = \sum_{C \in \Gamma_{i-1}} \size{C \cap T_{\alpha + i}}$.
\end{proof}

We are left with only the proof of \Cref{cl:const-P2}.

\begin{proof}[Proof of \Cref{cl:const-P2}]
We prove
\begin{enumerate}
\item[(i)]  $\pr\left(Z_C \leq \frac{3}{4} \size{C \cap T_{\alpha+i-1}}\right) \geq \pr(Z_C=1)= \frac{\size{C \cap T_{\alpha+i-1}}}{2^{\size{C \cap T_{\alpha + i -1}}}}$;
\item[(ii)] $\pr \left(Z_C \leq \frac{3}{4} \size{C \cap T_{\alpha+i-1}}\right) \geq 1- \frac{1}{2^{\Omega(\size{C \cap T_{\alpha + i -1}})}}$.
\end{enumerate}
From the above two statements, we are done with the claim by setting $\kappa$ as follows, which is $\Omega(1)$.

$$\max \left \lbrace \frac{\size{C \cap T_{\alpha+i-1}}}{2^{\size{C \cap T_{\alpha + i -1}}}},1- e^{-\Omega(\size{C \cap T_{\alpha + i -1}})} \right \rbrace .$$

Note that \malg calls $\mbox{\sas} \left(T_{\alpha+i-1},\frac{1}{2},r\right)$ in {\bf Phase 2.i}. 

In $\mbox{\sas} \left(T_{\alpha+i-1},\frac{1}{2},r\right)$, let $H_i \subseteq T_{\alpha + i -1}$ be the set of hubs sampled, where each point in $T_{\alpha + i -1}$ is (independently) included in $H_i$  with probability ${1}/{2}$.

For (i), $\pr\left(Z_C \leq \frac{3}{4} \size{C \cap T_{\alpha+i-1}}\right) \geq \pr(Z_C=1)$ is direct as $\size{C \cap T_{\alpha+i-1}} \geq 2$. From \Cref{lem:samp-add}, $Z_C=\size{C \cap T_{\alpha +i}}=1$ if $\size{H_i}=1$. So,
\begin{align*}
    \pr(Z_C=1) &= \frac{\size{C \cap T_{\alpha+i-1}}}{2^{\size{C \cap T_{\alpha + i -1}}}}
    \enspace.
\end{align*}

Now, we will prove (ii). From \Cref{lem:samp-add}, $Z_C =\size{C \cap T_{\alpha+i}} = \size{{C \cap H_i}}$ if $\size{{C \cap H_i}}>0$. Observe that $\pr({\size{C \cap H_i}} =0)=\frac{1}{2^{\size{C \cap T_{\alpha + i -1}}}}$. The expected number of points in ${C \cap H_i}$ is $\frac{\size{C \cap T_{\alpha + i -1}}}{2}$. Using Chernoff bound (\Cref{lem:cher}), $\pr\left(\size{{C \cap H_i}} > \frac{3}{4}\size{C \cap T_{\alpha + i -1}}\right) \leq e^{-\Omega\left(\size{C \cap T_{\alpha + i -1}}\right)}.$ Hence, putting things together,

\begin{align*}
    \pr\left(Z_C \leq \frac{3}{4}\size{C \cap T_{\alpha + i -1}}\right)
        &\geq \pr\left(Z_C \leq \frac{3}{4}\size{C \cap T_{\alpha + i -1}} \;\land\; \size{C \cap H_i} > 0\right) \\
        &\geq \pr\left(\size{C \cap H_i} \leq \frac{3}{4}\size{C \cap T_{\alpha + i -1}} \;\land\; \size{C \cap H_i} > 0\right) \\
        &\geq 1 - \pr\left(\size{C \cap H_i} > \frac{3}{4}\size{C \cap T_{\alpha + i -1}}\right) 
                - \pr\left(\size{C \cap H_i} = 0\right)
                && \text{(by the union bound)} \\
        &\geq 1 - \frac{1}{e^{\Omega\left(\size{C \cap T_{\alpha + i -1}}\right)}} 
                - \frac{1}{2^{\size{C \cap T_{\alpha + i -1}}}} \\
        &\geq 1 - \frac{1}{2^{\Omega\left(\size{C \cap T_{\alpha + i -1}}\right)}}
        \enspace.
\end{align*}

%
\end{proof}
\remove{This is because $t_\alpha=\Theta\left(\log ^{
(\alpha)} n  \cdot \left(\log ^{
(\alpha+1)}n \right)^{d+2}\right)$ and $\beta_d=\Theta\left(\log ^{(\alpha +1)} n\right)$.}


\section{Conclusion and Discussion}
\label{sec:conclude}

In this paper we show that even for large values of $k$, the classic $k$-center clustering problem in low-dimensional Euclidean space can be efficiently and very well approximated in the parallel setting of low-local-space \mpc. While some earlier works (see, e.g., \cite{coreset_kcenter1,coreset_kcenter2,coreset_kcenter3}) were able to obtain constant-round \mpc algorithms, they were relying on a large local space $\lspace \gg k$ allowing to successfully apply the core-set approach, which permits only limited communication. On the other hand, the low-local-space setting considered in this paper seems to require extensive communication between the machines to achieve any reasonable approximation guarantees. Therefore we believe (without any evidence) that the number of rounds of order $\Oh(\log\log n)$ may be almost as good as it gets. Also, we concede that our algorithm does not achieve a constant approximation guarantee, but we feel the approximation bound of $\Oh(\log^*n)$ is almost as good. Finally, our algorithm does not resolve the perfect setting of the $k$-center clustering in that it allows in the solution slightly more centers, $k + o(k)$ centers. Improving on these three parameters is the main open problem left by our work.  Additionally, due to our implementation of LSH in \mpc, we obtain a flexible  guarantee on global space. Reducing the value of $\rho$ decreases the amount of global space used by the algorithm (global space used is $\widetilde{O}(n^{1 + \rho})$) while increasing the constant inside approximation ratio $O(\log^\star n)$. 

We believe that solely using the technique in this paper, improving the approximation factor and/or number of rounds may not be possible (a detailed explanation is in tha paragraph below), but the approach may be useful for related problems in \mpc or other models. We remark that the extra space in global space complexity is mainly due to the use of LSH; note that, even in the RAM model setting, the use of LSH requires some extra space. However, we believe our implementation of LSH in \mpc could potentially see further applications, e.g., for other geometric problems.

Now, we discuss the limitation of our approach in designing algorithms with an improved approximation factor and/or number of rounds. Our algorithm is an iterative algorithm that refines the set of centers starting with the entire point set as the set of centers. As with many distributed algorithms, the iterative approach usually does not lead to a constant round algorithm. In particular, roughly speaking, our algorithm first samples points with inverse polynomial probability and then increases probability in a square root fashion. So, to go from $n$ centers to $k$ or even $\Oh(k)$ centers, one needs $\Omega(\log\log n))$ rounds. Increasing the probability along a faster schedule is unlikely to allow us to sufficiently bound the number of centers while increasing the probability along a slower schedule will increase the running time. Note also that the cost of the solution gets added over the rounds in our algorithm. The way our algorithm works, the cost of the solution in the first phase is $\Oh(r)$. After that, each optimal cluster has at most $\widetilde{\Oh}(\log^*n)$ centers. In phase 2, we apply \sas for $\Oh(\log^*n)$ rounds with $r$ as the radius that leads to approximation ratio $\Oh(\log^*n)$.  One may think to apply \sas in Phase 2 with a radius parameter less than $r$ in Phase 2. But in that case, guaranteeing the total number of centers to be $k+o(k)$ seems unlikely.



{Finally, we discuss the limitations of our technique in designing algorithms for $k$-center beyond constant-dimensional Euclidean spaces. The first bottleneck is the reliance on locally sensitive hashing. The second bottleneck arises in bounding \( z \) in \Cref{cl:X_C}. It is not difficult to observe that if a metric space has bounded doubling dimension, then \( z \) is also bounded in that space. A metric space has \emph{bounded doubling dimension} if there exists a constant \( \lambda \) such that any ball of radius \( r \) can be covered by at most \( \lambda \) balls of radius \( r/2 \). Consequently, our technique extends to any metric space that admits an efficient family of locally sensitive hash functions and has bounded doubling dimension.}


In this paper, we advance the understanding of the $k$-center problem in the \mpc model. Building on the work of Bateni et al.~\cite{bateni-kcenter}, we design a low-local-space \mpc algorithm that produces $k(1+o(1))$ centers with an $\mathcal{O}(\log^* n)$-approximation in $\mathcal{O}(\log\log n)$ rounds. We also acknowledge that a subsequent work~\cite{DBLP:conf/icalp/CzumajG0J25} has further improved these results. In the constant-dimensional Euclidean setting, fully scalable \mpc algorithms now achieve constant-round $(2+\varepsilon)$-approximation reporting exactly $k$ centers, and even $(1+\varepsilon)$-approximation using $(1+\varepsilon)k$ centers. In the high-dimensional regime, constant-round fully scalable \mpc algorithms achieve an $\Oh(\log n / \log \log n)$-approximation. These developments illustrate ongoing progress in both approximation quality and scalability for parallel $k$-center clustering.

An interesting open question is whether similar results can be achieved for the related $k$-means and $k$-medians problems in the low-local-space \mpc model when $k$ is large. In particular, it remains unclear if constant-round, fully scalable algorithms with strong approximation guarantees can be designed for these clustering objectives, analogous to the results for $k$-center.

\section*{Acknowledgments}
The authors gratefully acknowledge the anonymous referees for their valuable comments and insightful suggestions.

\bibliographystyle{alpha}
\bibliography{reference}

\appendix

\section*{Appendix}

\section{Concentration inequalities}
\label{sec:conc}

In our analysis we use some basic and standard concentration inequalities, which we present here for the sake of completeness.

\begin{lem}[{\bf  Chernoff bound}]
\label{lem:cher}
Let $X_1, \dots, X_n$ be independent random variables such that $X_i \in [0,1]$. For $X = \sum_{i=1}^n X_i$ and $ \E[X] \leq \mu_h$, the following hold for any $\eps \in (0,1)$.
\begin{align*}
    \pr \left( X \geq (1+\eps)\mu_h \right) &\leq e^{-\mu_h \eps^2/3}
    \enspace.
\end{align*}
\end{lem}

\begin{lem}[{\bf Hoeffding bound}]
\label{lem:hoeff}
Let $X_1, \dots, X_n$ be independent random variables such that $a_i \leq X_i \leq b_i$ and $X = \sum_{i=1}^n X_i$. Then, for all $t >0$,
\begin{align*}
    \pr\left({X \geq \E[X]} + t \right) &\leq
        \exp\left({-2t^2}/ {\sum_{i=1}^{n}(b_i-a_i)^2}\right)
    \enspace.
\end{align*}

\end{lem}
\end{document}

%% file: macro-gen-prev.tex
\usepackage{multirow}

\usepackage{hyperref}
\usepackage{microtype}
\usepackage{epsfig}
\usepackage{microtype}
\usepackage{graphicx}
\usepackage[linesnumbered,ruled, lined]{algorithm2e}
\usepackage{algpseudocode}
\usepackage{epsfig,enumerate,amsmath,amsfonts,amssymb,amsthm,mathrsfs,ifpdf}
\usepackage{cleveref}
\usepackage{indentfirst,relsize}
\usepackage{setspace}
\usepackage{enumerate}
\usepackage{enumitem}
\usepackage{latexsym}
\usepackage{stackrel}
\usepackage[all]{xy}
\usepackage[margin = 2.5cm]{geometry}
\usepackage[usenames,dvipsnames]{pstricks}
\usepackage{pst-grad} 
\usepackage{pst-plot} 
\usepackage{xspace}

\newcommand{\size}[1]{\left| #1 \right|}

\newcommand{\E}{\mathbb{E}}

\newcommand{\remove}[1]{}

\newcommand{\R}{\mathbb{R}}

\newcommand{\N}{\mathbb{N}}

\newcommand{\cC}{\mathcal{C}}

\newcommand{\cA}{\mathcal{A}}

\newcommand{\cQ}{\mathcal{Q}}

\newcommand{\Oh}{\mathcal{O}}
\newcommand{\tOh}{\widetilde{\mathcal{O}}}

\newcommand{\cF}{\mathcal{F}}

\newcommand{\cH}{\mathcal{H}}

\newcommand{\pr}{\mbox{Pr}}
\newcommand{\eps}{\epsilon}
\newcommand{\vareps}{\varepsilon}

\DeclareMathOperator*{\argmin}{arg\,min}
\newcommand{\complain}[1]{\textcolor{red}{#1}}

\theoremstyle{plain}
\newtheorem{theo}{Theorem}[section]
\newtheorem{lem}[theo]{Lemma}
\newtheorem{pro}[theo]{Proposition}

\newtheorem{cl}[theo]{Claim}

\theoremstyle{definition}
\newtheorem{defi}[theo]{Definition}
\newtheorem{rem}[theo]{Remark}
\newtheorem{obs}[theo]{Observation}

\renewcommand{\eps}{\vareps} 

%% file: macro-this-prev.tex
\newcommand{\kcen}{$k$-{center}\xspace}

\usepackage{mathtools}

\newcommand{\mpc}{\textsc{MPC}\xspace}
\newcommand{\specifiedmpc}{\mpc}

\newcommand{\sas}{{\sc Sample-And-Solve}\xspace}
\newcommand{\greedy}{{\sc Greedy}\xspace}
\newcommand{\ukc}{{\sc Uniform}-{\sc Center}\xspace}
\newcommand{\malg}{{\sc Ext}-$k$-{\sc Center}\xspace}
\newcommand{\cost}{\mbox{{\sc COST}}\xspace}
\renewcommand{\cost}{\textsc{cost}\xspace}
\newcommand{\opt}{\mbox{{\sc OPT}}\xspace}
\renewcommand{\opt}{\textsc{opt}\xspace}

\newcommand{\nns}{{\sc Nearest-Hub-Search}\xspace}
\newcommand{\close}{\mbox{close}\xspace}
\newcommand{\lspace}{\ensuremath{\mathcal{S}}\xspace}
\newcommand{\gspace}{\ensuremath{\mathcal{T}}\xspace}
\newcommand{\machines}{\ensuremath{\mathcal{M}}\xspace}
\renewcommand{\lspace}{\ensuremath{\textcolor[rgb]{0.50,0.00,0.00}{\mathfrak{s}}}\xspace}
\renewcommand{\gspace}{\ensuremath{\textcolor[rgb]{0.50,0.00,0.00}{\mathfrak{g}}}\xspace}
\renewcommand{\machines}{\ensuremath{\textcolor[rgb]{0.50,0.00,0.00}{\mathfrak{m}}}\xspace}
\renewcommand{\lspace}{\ensuremath{{\mathrm{s}}}\xspace}
\renewcommand{\gspace}{\ensuremath{{\mathrm{g}}}\xspace}
\renewcommand{\machines}{\ensuremath{{\mathrm{m}}}\xspace}

\newcommand{\range}{\ensuremath{\text{range}}}

\newcommand{\etal}{et al.\ }

\newcommand{\junk}[1]{}

\newcommand{\edit}[1]{\textcolor{black}{#1}} 

%% file: abstract.tex
    \begin{abstract}
We consider the classic $k$-center problem {in the constant dimensional Euclidean space} under a parallel setting, on the low-local-space Massively Parallel Computation (MPC) model, with local space per machine of $\mathcal{O}(n^{\delta})$, where $\delta \in (0,1)$ is an arbitrary constant. As a central clustering problem, the $k$-center problem has been studied extensively. Still, until very recently, all parallel MPC algorithms have been requiring $\Omega(k)$ or even $\Omega(k n^{\delta})$ local space per machine. While this setting covers the case of small values of $k$, for a large number of clusters these algorithms require large local memory, making them poorly scalable. The case of large $k$, $k \ge \Omega(n^{\delta})$, has been considered recently for the low-local-space MPC model by Bateni et al.\ (2021), who gave an $\mathcal{O}(\log \log n)$-round \mpc algorithm that produces $k(1+o(1))$ centers whose cost has multiplicative approximation of $\mathcal{O}(\log\log\log n)$. In this paper we extend the algorithm of Bateni et al. and design a low-local-space MPC algorithm that in $\mathcal{O}(\log\log n)$ rounds returns a clustering with $k(1+o(1))$ clusters that is an $\mathcal{O}(\log^*n)$-approximation for $k$-center.
\end{abstract}